\RequirePackage{fix-cm}
\documentclass[smallextended]{svjour3}       
\smartqed  
\usepackage{graphicx}
\usepackage[T1]{fontenc}
\usepackage[ruled]{algorithm2e}

\SetAlFnt{\small}
\SetAlCapFnt{\small}
\SetAlCapNameFnt{\small}
\SetAlCapHSkip{0pt}
\IncMargin{-\parindent}

\usepackage{etex}
\usepackage{graphics}
\usepackage{epsfig,verbatim}
\usepackage{arcs}
\usepackage{fontenc}
\usepackage{todonotes}
\usepackage{alltt}
\usepackage{amssymb,amsmath,amsfonts,epsf,url}
\usepackage{graphicx}
\usepackage{pstricks,pstricks-add,xcolor}
\usepackage{pst-node}
\usepackage{pst-coil}
\usepackage{tikz}
\usetikzlibrary{arrows,positioning,shapes,fit,calc,snakes,decorations.pathmorphing}
\tikzset{snake it/.style={decorate, decoration=snake}}

\pgfdeclarelayer{background}
\pgfsetlayers{background,main}

\usepackage{subfigure}

 \newcommand{\RRR}{\mathcal{R}}
 
\newcommand{\VVV}{\mathcal{V}}
\newcommand{\Oh}{\mathcal{O}}
\newcommand{\Ostar}{\mathcal{O^*}}
\newcommand{\FPT}{\mathcal{FPT}}
\newcommand{\APX}{\mathcal{APX}}
\newcommand{\NP}{\mathcal{NP}}

\newcommand{\nat}{\mathbb{N}}

\newcommand{\algo}{\textbf{{\textsc{Flip Distance-Algorithm}}}}
\newcommand{\algocomb}{\textbf{{\textsc{Combinatorial Flip Distance-Algorithm}}}}
\newcommand{\rec}{\textbf{{\textsc{Recursive-Construction(${\cal S}$)}}}}
\newcommand{\paramproblem}[4]{\noindent {\sc #1}
\\
{\bf Given:} #2\\
{\bf Parameter:} #3\\
{\bf Question:} #4}

\def\eg{{\em e.g.}}

\def\ie{{\em i.e.}}

\def\CFD{{\sc Combinatorial Flip Distance}}
\def\FD{{\sc Flip Distance}}

\newlength{\alginputwidth}
\newlength{\algboxwidth}

\newcommand{\algtitle}[1]{\underline{{\bf #1}} \vspace*{1mm}\\}

\newsavebox{\algbox}
\newsavebox{\captionbox}
\newenvironment{algorithmnew}[2]%
    {
        \setlength{\algboxwidth}{\columnwidth}
        \addtolength{\algboxwidth}{-\columnsep}
        \addtolength{\algboxwidth}{-1mm}
        \setlength{\alginputwidth}{\algboxwidth}
        \addtolength{\alginputwidth}{-1.7cm}
        \begin{figure}[htbp]
            \vspace*{-1mm}
            \centering
            \begin{lrbox}{\captionbox}
                \begin{minipage}[b]{\algboxwidth}
                    \centering
                    \caption{#1}
                    \label{#2}
                \end{minipage}
            \end{lrbox}
            \begin{lrbox}{\algbox}
                \begin{minipage}[b]{\algboxwidth}
                    \footnotesize
                    \vspace*{2mm}
    } 
    {
                    \vspace*{0.2mm}
               \end{minipage}
            \end{lrbox}
            \fbox{\usebox{\algbox}\hspace*{1mm}}
            \usebox{\captionbox}
            \vspace*{-1mm}
      \end{figure}
    }
\newsavebox{\algcodebox}
\newenvironment{codeblock}%
    {
        \begin{enumerate}
            \setlength{\itemsep}{2pt}
            \setlength{\parsep}{0pt}
            \setlength{\topsep}{0pt}
            \setlength{\parskip}{0pt}
            \setlength{\partopsep}{0pt}
    } 
    {\end{enumerate}}
\newcommand{\step}{\item}

\begin{document}

\title{Computing the flip distance between triangulations\thanks{An extended abstract of this work (without the complete proofs and the extension to labeled triangulated graphs) appears in~\cite{kanjxiastacs}.}}


\author{Iyad Kanj \and Eric Sedgwick \and Ge Xia
}

\institute{I. Kanj (corresponding author) \at
              School of Computing, DePaul University, Chicago, USA. \\
              Tel.: +1312-362-5558\\
              Fax: +1312-362-6118\\
              \email{ikanj@cs.depaul.edu}           
           \and
           E. Sedgwick  \at
           School of Computing, DePaul University, Chicago, USA. \\
           Tel.: +1312-362-5558\\
           Fax: +1312-362-6118\\
           \email{esedgwick@cs.depaul.edu}
           \and
           G. Xia \at
           Department of Computer Science, Lafayette College, Easton, USA. \\
           Tel.: +1610-330-5415 \\
           Fax: +1610-330-5059 \\
           \email{xiag@lafayette.edu}
}

\date{ }

\maketitle

\begin{abstract}
Let ${\cal T}$ be a triangulation of a set ${\cal P}$ of $n$ points in the plane, and let $e$ be an edge shared by two triangles in ${\cal T}$ such that the quadrilateral $Q$ formed by these two triangles is convex. A {\em flip} of $e$ is the operation of replacing $e$ by the other diagonal of $Q$ to obtain a new triangulation of ${\cal P}$ from ${\cal T}$. The {\em flip distance} between two triangulations of ${\cal P}$ is the minimum number of flips needed to transform one triangulation into the other. The {\sc Flip Distance} problem asks if the flip distance between two given triangulations of ${\cal P}$ is at most $k$, for some given $k \in \nat$. It is a fundamental and a challenging problem.

We present an algorithm for the {\sc Flip Distance} problem that runs in time $\Oh(n + k \cdot c^{k})$, for a constant~$c \leq 2 \cdot 14^{11}$, which implies that the problem is fixed-parameter tractable. We extend our results to triangulations of polygonal regions with holes, and to labeled triangulated graphs.
\keywords{Flip distance \and Triangulations \and Triangulated graphs \and Parameterized complexity}
\end{abstract}
\section{Introduction} \label{sec:intro}
Let ${\cal P}$ be a set of $n$ points in the plane. A {\em triangulation} of ${\cal P}$ is a partitioning of the convex hull of ${\cal P}$ into triangles such that the set of vertices of the triangles in the triangulation is ${\cal P}$. Note that the convex hull of ${\cal P}$ may contain points of ${\cal P}$ in its interior.

A {\em flip} to an (interior) edge $e$ in a triangulation of ${\cal P}$ is the operation of replacing $e$ by the other diagonal of the quadrilateral formed by the two triangles that share $e$, provided that this quadrilateral is convex; otherwise, flipping $e$ is not admissible. The {\em flip distance} between two triangulations ${\cal T}_{initial}$ and ${\cal T}_{final}$ of ${\cal P}$ is the length of a shortest sequence of flips that transforms ${\cal T}_{initial}$ into ${\cal T}_{final}$. This distance is always well-defined and is $\Oh(|{\cal P}|^2)$ (\eg, see~\cite{urrutia}). The {\sc Flip Distance} problem is: Given two triangulation ${\cal T}_{initial}$ and ${\cal T}_{final}$ of ${\cal P}$, and $k \in \nat$, decide if the flip distance between ${\cal T}_{initial}$ and ${\cal T}_{final}$ is at most $k$.

Triangulations are a very important subject of study in computational geometry, and they have applications in computer graphics, visualization, and geometric design (see~\cite{triangulations,okabe,saalfeld,schumaker,watson}, to name a few). Flips in triangulations and the {\sc Flip Distance} problem have received a large share of attention (see~\cite{bosehurtado} for a review). The {\sc Flip Distance} problem is a very fundamental and challenging problem, and different aspects of this problem have been studied, including the combinatorial, geometrical, topological, and computational aspects~\cite{hurtadolower,mulzer,bosehurtado,cleary,hanke,urrutia,lawson,lubiw,pilz,sibson,sleator}. Lawson~\cite{lawson} studied flips in triangulations, and proved that any two triangulations of ${\cal P}$ can be transformed into one another by a finite sequence of flips; an analysis of Lawson's result~\cite{lawson} yields an $O(n^2)$ upper bound on the number of flips needed to transform one triangulation into another, and hence, on the flip distance between any two triangulations. We can define the triangulations graph of ${\cal P}$, whose vertex-set is the set of all triangulations of ${\cal P}$, and in which two triangulations/vertices are adjacent if and only if their distance is 1. Lawson's result~\cite{lawson} implies that the triangulations graph has diameter $\Oh(n^2)$. Moreover, it is known that the number of vertices in the triangulations graph is $\Omega(2.631^n)$~\cite{triangulationslowerbound}. Therefore, solving the {\sc Flip Distance} problem by finding a shortest path between the two triangulations in the triangulations graph is not feasible.

The complexity of the {\sc Flip Distance} problem was resolved very recently (2012) by Lubiw and Pathak~\cite{lubiw} who showed the problem to be $\NP$-complete. Simultaneously, and independently, the problem was shown to be $\APX$-hard by Pilz~\cite{pilz}. Very recently, Aichholzer et al.~\cite{mulzer} showed the problem to be $\NP$-complete for triangulations of a simple polygon. Resolving the complexity of the problem for the special case when ${\cal P}$ is in a convex position (\ie, triangulations of a convex polygon) is a long-standing open problem (see~\cite{sleator}); this problem is equivalent to the problem of computing the rotation distance between two rooted binary trees~\cite{cleary,sleator}. Cleary and St.~John~\cite{cleary} showed that this special case (convex polygon) is fixed-parameter tractable ($\FPT$): They gave a kernel of size $5k$ for the problem and presented an $\Oh^*((5k)^{k})$-time  $\FPT$ algorithm based on this kernel (the $\Ostar()$ notation suppresses polynomial factors in the input size). The upper bound on the kernel size for the convex case was subsequently improved to $2k$ by Lucas~\cite{lucas}, who also gave an $\Oh^*(k^{k})$-time $\FPT$ algorithm for this case. The kernelization approaches used in~\cite{cleary,lucas} for the convex case are not applicable to the general case. In particular, the reduction rules used in~\cite{cleary,lucas} to obtain a kernel for the convex case, and hence the $\FPT$ algorithms based on these kernels, do not generalize to the problems under consideration in this paper.

In this paper we present an $\Oh(n + k \cdot c^{k})$-time algorithm, where $c \leq 2 \cdot 14^{11}$, for the {\sc Flip Distance} problem for triangulations of an arbitrary point-set in the plane, which shows that the problem is $\FPT$.
Our result is a significant improvement over the $\Oh^*(k^k)$-time algorithm by Lucas~\cite{lucas} for the convex case, which is a special case of the {\sc Flip Distance} problem.
While it is not very difficult to show that the {\sc Flip Distance} problem is $\FPT$ based on some of the structural results in this paper, obtaining an $\Ostar(c^{k})$-time algorithm, for some constant $c$, is quite involved.

Our approach is as follows. For any solution to a given instance of the problem, we can define a directed acyclic graph (DAG), whose nodes are the flips in the solution, that captures the dependency relation among the flips. We show that any topological sorting of this DAG corresponds to a valid solution of the instance. The difficult part is how, without knowing the DAG, to navigate the triangulations and perform the flips in an order that corresponds to a topological sorting of the DAG. The key is to show that there exists a sequence of ``flip/move''-type local actions in the triangulations, where each local action has constantly-many choices, that corresponds to a topological sorting of the DAG associated with a solution to the instance, and such that the length of this sequence is linear in the number of nodes in the DAG. This enables us to present a search-tree algorithm that searches for such a sequence in $\Ostar(c^k)$ time. To achieve the above goal, we develop results that reveal some of the structural intricacies of this fundamental and challenging problem.

Our approach and techniques are very generic. They not only work seamlessly for other types of triangulations in the geometric setting, but they can also be adapted to work for triangulated graphs in the combinatorial setting. For the combinatorial setting in particular, significant work has been done on the flip distance of triangulated graphs. We highlight below some of this work.

Wagner~\cite{wagner} studied the flip distance of triangulated graphs, which are maximal (simple) planar graphs. All the faces of an embedding of a triangulated graph, including the external face, are triangles (cycles of length 3). A \emph{flip} of an edge in this setting is the operation of replacing the edge with another edge to obtain another triangulated graph, provided that the obtained graph remains simple. Wagner~\cite{wagner} showed an $\Oh(n^2)$ upper bound on the flip distance of any two triangulated graphs by showing that any triangulated graph can be transformed into a canonical triangulated graph using $\Oh(n^2)$ flips. A sequence of improvements on this upper bound followed, and the current-best upper bound due to Mori {\em et al.}~\cite{mori} is $6n-30$. The flip distance between vertex-labeled triangulated graphs was also studied~\cite{bosehurtado,hurtado,tarjan}; those are triangulated graphs defined on the same vertex set, or given with a bijection between their two vertex sets. It was shown that the flip distance between any two (vertex) labeled triangulated graphs is $\Oh(n\lg{n})$~\cite{tarjan}, and this upper bound is tight; that is, there exist labeled triangulated graphs whose flip distance is $\Omega(n\lg{n})$. The (classical) complexity of computing the flip distance for both triangulated graphs and labeled triangulated graphs remains open.

In Section~\ref{sec:extensions}, we show how the presented algorithm for {\sc Flip Distance} can be employed to compute the flip distance between triangulations of any polygonal region, even with holes in its interior.  We also show in Section~\ref{sec:extensions} how to extend our results to compute the flip distance between two labeled triangulated graphs in time $\Oh(n+k\cdot c^{k}\cdot \lg{n})$.

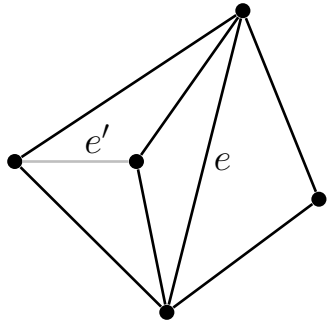
\begin{figure}\label{fig:admissible}
\begin{center}
\begin{tikzpicture}
[vertex/.style={fill,circle,inner sep = 2pt}]

\node (x) at (7, 1) [vertex]{};
\node (a) at (4,-1) [vertex] {};
\node (b) at (5.6,-1) [vertex]{};
\node (c) at (8,-1.5) [vertex] {};
\node (d) at (6,-3) [vertex] {};

\draw[line width=1] (a) -- (x) -- (c) -- (d) -- (a);

\draw[black, line width=1] (x) -- (d) node[midway, anchor=west] {\Large $e$};
\draw[black, line width=1] (x) -- (b) ;
\draw[black, line width=1] (b) -- (d) ;
\draw[lightgray, line width=1] (b) -- (a) node[midway, anchor=south west] {\black \Large $e'$};
\end{tikzpicture}

\end{center}
\caption{Illustration of admissible vs.~inadmissible flips: flipping $e$ is admissible but flipping $e'$ is inadmissible.}
\end{figure}

\section{Preliminaries}\label{sec:prelim}

\paragraph{{\bf Flips, triangulations, and flip distance.}} Let ${\cal P}$ be a set of $n$ points in the plane, and let ${\cal T}$ be a triangulation of ${\cal P}$. Let $e$ be an interior (non-boundary) edge in ${\cal T}$. The {\em quadrilateral associated with} $e$ in ${\cal T}$, denoted $Q_e$, is defined to be the quadrilateral formed by the two adjacent triangles in ${\cal T}$ that share $e$ as an edge. Let $e$ be an edge in ${\cal T}$ such that the quadrilateral $Q_e$ in ${\cal T}$ associated with $e$ is convex.  A {\em flip} $f$ with {\em underlying edge} $e$ is an operation performed to $e$ in triangulation ${\cal T}$ that removes $e$ and replaces it with the other diagonal of $Q_e$, thus obtaining a new triangulation of ${\cal P}$ from ${\cal T}$. We use the notation $\epsilon(f)$ to denote the underlying edge $e$ of a flip $f$ in ${\cal T}$, and the notation $\phi(f)$ to denote the new edge resulting from flip $f$. Note that $\phi(f)$ is not in ${\cal T}$. We say that a flip to an edge $e$ is {\em admissible} in triangulation ${\cal T}$ if $e$ is in ${\cal T}$ and the quadrilateral $Q_e$ associated with $e$ is convex. (See Fig.~\ref{fig:admissible} for illustration.) We say that two distinct edges $e$ and $e'$ in ${\cal T}$ {\em share a triangle} if $e$ and $e'$ appear in the same triangle in ${\cal T}$. We say that two distinct edges $e$ and $e'$ between points in ${\cal P}$ {\em cross} if $e$ and $e'$ intersect in their interior.

Let ${\cal T}$ be a triangulation. A sequence of flips $F=\langle f_1,\ldots, f_r \rangle$ is {\em valid} with respect to ${\cal T}$ if there exist triangulations ${\cal T}_0, \ldots, {\cal T}_r$ such that ${\cal T}_0 = {\cal T}$, $f_i$ is admissible in ${\cal T}_{i-1}$, and performing flip $f_i$ in ${\cal T}_{i-1}$ results in triangulation ${\cal T}_i$, for $i=1, \ldots, r$. In this case we say that ${\cal T}_r$ is the {\em outcome} of applying $F$ to ${\cal T}$ and that $F$ {\em transforms} ${\cal T}$ into ${\cal T}_r$, and we write ${\cal T} \xrightarrow{F} {\cal T}_r$. The {\em length} of $F$, denoted $|F|$, is the number of flips in it. Many flips in a sequence $F$ may have the same underlying edge, but all those flips are distinct flips. For two flips $f_i$ and $f_h$ of $F$ such that $i < h$, a flip $f_p$ in $F$ is said to be {\em between} $f_i$ and $f_h$ if $i < p < h$.

For two triangulations ${\cal T}_{initial}$ and ${\cal T}_{final}$ of ${\cal P}$, the {\em flip distance} between ${\cal T}_{initial}$ and ${\cal T}_{final}$ is the smallest $d \in \nat$ such that there is a sequence $F$ of length $d$ satisfying that ${\cal T}_{initial} \xrightarrow{F} {\cal T}_{final}$.
The {\sc Flip Distance} problem is defined as follows: \\

\paramproblem{{\sc Flip Distance}}{Two triangulation ${\cal T}_{initial}$ and ${\cal T}_{final}$ of ${\cal P}$.}{$k$.}{Is the flip distance between ${\cal T}_{initial}$ and ${\cal T}_{final}$ at most $k$?} \\

Let $({\cal T}_{initial}, {\cal T}_{final}, k)$ be an instance of {\sc Flip Distance}. A \emph{solution} to $({\cal T}_{initial}, {\cal T}_{final}, k)$ is any sequence of flips $F$ (if it exists) such that ${\cal T}_{initial} \xrightarrow{F} {\cal T}_{final}$ and $|F| \leq k$. A solution is {\em minimum} if its length is minimum over all solutions to $({\cal T}_{initial}, {\cal T}_{final}, k)$. Observe that if a minimum solution to $({\cal T}_{initial}, {\cal T}_{final}, k)$ exists, then its length is equal to the flip distance between ${\cal T}_{initial}$ and ${\cal T}_{final}$.

\paragraph{{\bf Parameterized complexity.}}
A {\em parameterized problem} is a set of instances of the form
$(x, k)$, where $x$ is the input instance and $k \in \nat$
is the {\it parameter}. A parameterized problem is
{\it fixed-parameter tractable}, shortly $\FPT$, if there
is an algorithm that solves the problem in time $f(k)|x|^{c}$, where $f$ is a computable function and $c > 0$ is
a constant. A parameterized problem is {\em kernelizable}
if there exists a polynomial-time reduction that maps an instance $(x,k)$ of
the problem to another instance $(x',k')$ such that: (1) $|x'| \leq \lambda(k)$ for
some computable function $\lambda$, (2) $k' \leq \lambda(k)$, and (3) $(x,k)$ is a yes-instance
of the problem if and only if $(x',k')$ is. The instance
$(x',k')$ is called the {\em kernel} of $(x, k)$. We refer to~\cite{fptbook,rolfbook} for more information about parameterized complexity.

\paragraph{{\bf Graphs.}}
All graphs considered in this paper are simple finite graphs with no multiple edges or self-loops. Let $G$ be a graph. $V(G)$ and $E(G)$ denote the vertex-set and the edge-set of $G$, respectively, and $|G|$ denotes the  {\em size} of $G$, which is $|V(G)| + |E(G)|$. If $S$ is any set of vertices in $G$, we write $G-S$ for the subgraph of $G$ obtained by deleting all the vertices in $S$ and their incident edges. For a directed graph $G$, a {\em weakly connected component} of $G$ is a (maximal) connected component of the underlying undirected graph of $G$ (that is, the undirected graph obtained from $G$ by replacing all directed edges with undirected edges); for simplicity, we will use the term {\em component} of a directed graph $G$ to refer to a weakly connected component of $G$.

A graph is {\it planar} if it can be drawn in the plane without edge
intersections (except at the endpoints). A {\it plane graph} has a
fixed drawing.  Each maximal connected region of the plane minus the
drawing is an open set; these are the {\em faces}. One is unbounded,
called the {\em outer face}.  A face is a {\em triangle}, or a {\em triangular face}, if it is a cycle of length 3 (\ie, $C_3$). A {\em rotation system}, $\RRR$, of a plane graph $G$ is a clockwise/counterclockwise cyclical ordering of the edges around each vertex of $G$. We sometime refer to the vertices of a plane graphs by {\em points}.

A simple planar graph $G$ is {\em triangulated} if $G$ can be embedded in the plane so that each face
of the embedding, including the external/outer face, is a triangle. It is well known that $G$ is triangulated if and only if it is a maximal planar graph, in which case $G$ has $3n-6$ edges.
If $G$ is a triangulated planar graph on $n \geq 4$ vertices then $G$ is 3-connected~\cite{diestel}, and it has a unique combinatorial embedding up to homeomorphism~\cite{diestel} (\ie, topological isomorphism). Moreover,
for any two triangulated planar graphs $G$ and $H$, every isomorphism between $G$ and $H$ is a homeomorphism~\cite{diestel}. We call an embedded triangulated planar graph a {\em plane triangulation}.  We refer to~\cite{diestel} for more information on graphs.

\section{Structural results}\label{sec:structural}
Let ${\cal T}$ be a triangulation and let $F=\langle f_1,\ldots, f_r \rangle$ be a valid sequence of flips with respect to ${\cal T}$. We denote by ${\cal T}_i$, for $i=1, \ldots, r$, the triangulation that is the outcome of applying the (valid) subsequence of flips $\langle f_1,\ldots, f_i \rangle$ to ${\cal T}$.

\begin{definition}\label{def:adjacent}\rm
Let $f_i$ and $f_j$ be two flips in $F$ such that $1 \leq i < j\leq r$. Flip $f_j$ is said to be {\em adjacent} to flip $f_i$, denoted $f_i \rightarrow f_j$, if:
\begin{itemize}
\item[(1)] either $\phi(f_i) = \epsilon(f_j)$ (\ie, $\epsilon (f_j)$ results from flip $f_i$), or $\phi(f_i)$ is an edge of the quadrilateral $Q_{\epsilon(f_j)}$ associated with $\epsilon(f_j)$ in triangulation ${\cal T}_{j-1}$; and
\item[(2)] $\phi(f_i)$ is not flipped between $f_i$ and $f_j$, that is, there does not exist a flip $f_p$ in $F$, where $i < p < j$, such that $\epsilon(f_p) =\phi(f_i)$.
\end{itemize}
\end{definition}

The above adjacency relation defined on the flips in $F$ can be naturally represented by a directed acyclic graph (DAG), denoted ${\cal D}_F$,
where the nodes of ${\cal D}_F$ are the flips in $F$, and its arcs represent the (directed) adjacencies in $F$. Note that by definition, if $f_i \rightarrow f_j$ then $i< j$. For simplicity, we will label the nodes in ${\cal D}_F$ with the labels of their corresponding flips in $F$.

\begin{lemma}\label{lem:dagsize}
Every node in ${\cal D}_F$ has indegree at most $5$. Therefore, $|E({\cal D}_F)| \leq 5 \cdot |V({\cal D}_F)|$ and $|{\cal D}_F| \leq 6 \cdot |V({\cal D}_F)|$.
\end{lemma}

\begin{proof}
For any node $f_j \in {\cal D}_F$, the quadrilateral $Q_{\epsilon(f_j)}$ associated with $\epsilon(f_j)$ in ${\cal T}_{j-1}$ has four edges.  Therefore, every node $f_j$ in ${\cal D}_F$ has indegree at most 5, taking into consideration the possible last flip that created $\epsilon(f_j)$. \qed \end{proof}

Recall that a topological sorting of a DAG is {\em any} ordering of its nodes that satisfies: For any two nodes $u, v$ in the DAG, if there is an arc from $u$ to $v$, then $u$ appears before $v$ in the ordering. There could be many different topological sorting of ${\cal D}_F$, but the following key lemma asserts that all of them yield the same outcome and have the same DAG:


\begin{lemma} \label{lem:ts}
Let ${\cal T}_0$ be a triangulation and let $F=\langle f_1,\ldots, f_r \rangle$ be a sequence of flips such that ${\cal T}_0  \xrightarrow{F} {\cal T}_r$. Let $\pi(F)$ be a permutation of the flips in $F$ such that $\pi(F)$ is a topological sorting of ${\cal D}_F$. Then $\pi(F)$ is a valid sequence of flips such that ${\cal T}_0 \xrightarrow{\pi(F)} {\cal T}_r$. Furthermore, the DAG ${\cal D}_{\pi(F)}$, defined based on the sequence $\pi(F)$, is the same directed graph as ${\cal D}_F$.
\end{lemma}

\begin{proof}
We will prove a stronger statement than that of the lemma. We define a marking on the edges of the triangulations resulting from a sequence of flips as follows: After a flip $f_i$ is performed, {\em mark} the resulting edge $\phi(f_i)$ by $f_i$; if the edge $\phi(f_i)$ is already marked by a previous flip, overwrite the previous mark by the new one.

We will prove that (1) $\pi(F)$ is a valid sequence of flips such that ${\cal T}_0 \xrightarrow{\pi(F)} {\cal T}_r$, (2) the marking on ${\cal T}_r$ after $\pi(F)$ is performed is the same as that after $F$ is performed (\ie, each edge in ${\cal T}_r$, if marked, is marked by the same flip), and (3) ${\cal D}_{\pi(F)} = {\cal D}_F$.

Proceed by induction on $|F|$. If $|F| \leq 1$, then obviously the statement holds true. Suppose that the statement is true for any $F$ such that $|F| < r$, where $r > 1$, and consider a sequence $F$ such that $|F|=r$.
Let $f_s$ be the last flip in $\pi(F)$. Since $\pi(F)$ is a topological sorting of ${\cal D}_F$, $f_s$ must be a sink in ${\cal D}_F$. It follows that no flip after $f_s$ in $F$ is adjacent to $f_s$ in ${\cal D}_F$. Let $Q_{\phi(f_s)}$ be the quadrilateral associated with $\phi(f_s)$ in triangulation ${\cal T}_s$. Then no flip after $f_s$ in $F$ has its underlying edge as a boundary edge of $Q_{\phi(f_s)}$ or as a diagonal of $Q_{\phi(f_s)}$, which means that the two adjacent triangles forming $Q_{\phi(f_s)}$ in ${\cal T}_s$ remain unchanged throughout the flips after $f_s$ in $F$. Therefore, we can safely move the flip $f_s$ to the end of the sequence $F$ without affecting the other flips in $F$ nor the validity of $F$. Let this new sequence be $F'$. From the previous arguments, it follows that $F'$ transforms ${\cal T}_0$ into ${\cal T}_r$, the marking on ${\cal T}_r$ resulting from $F'$ is the same as the marking on ${\cal T}_r$ resulting from $F$, and ${\cal D}_F = {\cal D}_{F'}$.

Since $f_s$ appears at the end of $F'$, $F'-f_s$ is a valid sequence with respect to ${\cal T}_0$ that transforms ${\cal T}_0$ into some triangulation ${\cal T}$ such that ${\cal T} \xrightarrow{f_s} {\cal T}_r$. Note that since $f_s$ is a sink in ${\cal D}_F$, $\pi(F)-f_s$ is a permutation of the flips in $F'-f_s$ that is a topological sorting of ${\cal D}_F -f_s$. By the inductive hypothesis, $\pi(F)-f_s$ transforms ${\cal T}_0$ into ${\cal T}$, the marking on ${\cal T}$ resulting from $F'-f_s$ is the same as that resulting from $\pi(F)-f_s$, and ${\cal D}_{F'-f_s} = {\cal D}_{\pi(F)-f_s}$.

Since ${\cal T} \xrightarrow{f_s} {\cal T}_r$, appending $f_s$ to the end of $\pi(F)-f_s$ results in $\pi(F)$ such that ${\cal T}_0 \xrightarrow{\pi(F)} {\cal T}_r$. Since $\phi(f_s)$ is marked by flip $f_s$ in ${\cal T}_r$ both after applying $F'$ and $\pi(F)$, the marking on ${\cal T}_r$ resulting from $\pi(F)$ is the same as that resulting from $F$. Now ${\cal D}_F$ is formed by adding $f_s$ to ${\cal D}_{F'-f_s}$ along with its incoming edges, and ${\cal D}_{\pi(F)}$ is formed by adding $f_s$ to ${\cal D}_{\pi(F)-f_s}$ along with its incoming edges. Because the marking on ${\cal T}$ resulting from $\pi(F)-f_s$ is the same as the marking resulting from $F'-f_s$, the incoming edges to $f_s$ in ${\cal D}_{\pi(F)}$ are the same as the incoming edges to $f_s$ in ${\cal D}_F$. It follows that ${\cal D}_{\pi(F)} = {\cal D}_{F'} = {\cal D}_{F}$. This completes the inductive proof. \qed
\end{proof}

\begin{corollary}\label{cor:ts}
Let ${\cal T}_0$ be a triangulation and let $F=\langle f_1,\ldots, f_r \rangle$ be a sequence of flips such that ${\cal T}_0  \xrightarrow{F} {\cal T}_r$. For any given ordering $(C_1, \ldots, C_{\ell})$ of the components in ${\cal D}_F$, there is a permutation $\pi(F)$ of the flips in $F$ such that ${\cal T}_0 \xrightarrow{\pi(F)} {\cal T}_r$, and such that for any two flips $f_i \in C_t$ and $f_j \in C_s$, where $1 \leq t < s \leq \ell$, $f_i$ appears before $f_j$ in $\pi(F)$. That is, all the flips in the same component appear as a consecutive ``block'' (\ie, the flips in the same component appear consecutively) in $\pi(F)$, and the order of the blocks of flips in $\pi(F)$ is the same as the given order of their corresponding components.
\end{corollary}

\begin{proof}
For any given ordering of the components in ${\cal D}_F$, there is a topological sorting of ${\cal D}_F$ in which all the flips in the same component of ${\cal D}_F$ appear consecutively as a block, and in which the blocks of flips appear in the same order as the given order of their components. \qed
\end{proof}

\begin{definition}\rm \label{def:changededges}
Let $({\cal T}_{initial}, {\cal T}_{final}, k)$ be an instance of {\sc Flip Distance}. An edge in ${\cal T}_{initial}$ that is not in ${\cal T}_{final}$ is called a {\em changed edge}. If a sequence $F$ is a solution to the instance $({\cal T}_{initial}, {\cal T}_{final}, k)$, we call a component in ${\cal D}_F$ {\em essential} if the component contains a flip $f$ such that $\epsilon(f)$ is a changed edge, otherwise, the component is called {\em nonessential}. (Recall that a component of ${\cal D}_F$ stands for a weakly connected component of ${\cal D}_F$.)
\end{definition}

\begin{lemma}\label{lem:essential} Let $({\cal T}_{initial}, {\cal T}_{final}, k)$ be an instance of {\sc Flip Distance}, and suppose that $F$ is a minimum solution to the instance. Then every component of ${\cal D}_F$ is essential.
\end{lemma}

\begin{proof}
Suppose, to get a contradiction, that ${\cal D}_F$ contains a nonessential component $C$. Let $F_C$ be the subsequence of $F$ consisting of the flips that are in $C$. We will show that $F - F_C$ is a solution to the instance $({\cal T}_{initial}, {\cal T}_{final}, k)$, which contradicts the minimality of $F$.

By Corollary~\ref{cor:ts}, we can assume that all the flips in $F_C$ appear consecutively (\ie, as a single block) at the end of $F$. Let ${\cal T'}$ be the outcome of applying $F - F_C$ to ${\cal T}_{initial}$. It suffices to show that ${\cal T'} = {\cal T}_{final}$. Suppose that this is not the case. Since the number of edges in ${\cal T'}$ and ${\cal T}_{final}$ is the same, there must exist an edge $e \in {\cal T'}$ such that $e \notin {\cal T}_{final}$.  Therefore, $C$ must contain a flip $f$ such that $\epsilon(f)=e$; assume that $f$ is the first such flip in $C$. Since $C$ is nonessential, $e \notin {\cal T}_{initial}$, otherwise $e$ would be a changed edge. Therefore, there must exist a flip $f'$ in $F- F_C$ such that $\phi(f')=e$; we can assume that $f'$ is the last such flip in $F- F_C$. By the definition of adjacency in ${\cal D}_F$, there is an arc from node $f'$ in ${\cal D}_F - C$ to node $f$ in $C$, contradicting the assumption that $C$ is a component of ${\cal D}_F$. \qed
\end{proof}

Let $({\cal T}_{initial}, {\cal T}_{final}, k)$ be an instance of {\sc Flip Distance}, and suppose that $F$ is a minimum solution for $({\cal T}_{initial}, {\cal T}_{final}, k)$. By Lemma~\ref{lem:essential}, ${\cal D}_F$ does not contain nonessential components, and by Corollary~\ref{cor:ts}, we can assume that all the flips in the same component of ${\cal D}_F$ appear as a consecutive block in $F$. We shall call such a solution $F$ satisfying the above properties, including minimality, a {\em normalized} solution. Suppose that $F=\langle f_1, \ldots, f_r \rangle$, where $r \leq k$, is a normalized solution to an instance $({\cal T}_{initial}, {\cal T}_{final}, k)$ of {\sc Flip Distance}, and let $C$ be a component of ${\cal D}_F$. The following lemmas provide several sufficient conditions for a directed path to exist between two flips in $C$.

\begin{lemma} \label{lem:crossing}
Let $f_i$ and $f_h$, where $i < h$ and $\epsilon(f_i) \neq \epsilon(f_h)$, be two flips in $C$ such that $\phi(f_h)$ crosses $\epsilon(f_i)$, and $\epsilon(f_i)$ is not flipped between $f_i$ and $f_h$. There is a directed path from $f_i$ to $f_h$ in $C$.
\end{lemma}
\begin{proof}
Assume that the statement is not true, and let $f_i$ and $f_h$ be the closest pair of flips in $C$ (in terms of the number of flips between $f_i$ and $f_h$) satisfying the conditions in the statement of the lemma and such that there is no directed path from $f_i$ to $f_h$ in $C$.

Consider the quadrilateral $Q_{\phi(f_h)}$ associated with $\phi(f_h)$ in ${\cal T}_h$. We first argue that $\epsilon(f_i)$ must cross an edge of  $Q_{\phi(f_h)}$. Because $\phi(f_h)$ crosses $\epsilon(f_i)$, there is a point $u$ in the plane that is interior to both $\epsilon(f_i)$ and $\phi(f_h)$, and hence, interior to $Q_{\phi(f_h)}$ as well. Since $\epsilon(f_i) \neq \epsilon(f_h)$ and $\phi(f_h)$ crosses $\epsilon(f_i)$, $\epsilon(f_i)$ cannot be a diagonal or an edge of $Q_{\phi(f_h)}$. The preceding statement, combined with the fact that the interior of $Q_{\phi(f_h)}$ is devoid of points of ${\cal P}$, implies that at least one endpoint of $\epsilon(f_i)$ lies outside of $Q_{\phi(f_h)}$. Now since $\epsilon(f_i)$ passes through $u$, which is interior to $Q_{\phi(f_h)}$, by the Jordan curve theorem, $\epsilon(f_i)$
intersects an edge $e$ (of the boundary) of $Q_{\phi(f_h)}$.

It follows from the fact that $\epsilon(f_i)$ intersects $e$ that there must exist a flip $f_p$, where $i \leq p < h$, such that $\phi(f_p)=e$. Since $e$ is an edge of $Q_{\phi(f_h)}$, which contains $\phi(f_h)$ as a diagonal in ${\cal T}_h$, $e$ is also an edge of $Q_{\phi(f_h)}$ in ${\cal T}_{h-1}$. This means that $\phi(f_p)$ and $\epsilon(f_h)$ share a triangle in ${\cal T}_{h-1}$, and hence there is an arc from $f_p$ to $f_h$ (we can assume that $f_p$ is the last flip before $f_h$ such that $e = \phi(f_p)$). If $i = p$ then there is a path (of length 0) from $f_i$ to $f_p$; on the other hand, if $i < p$, then since $\phi(f_p)$ crosses $\epsilon(f_i)$, $p < h$, and $\epsilon(f_p) \neq \epsilon(f_i)$, by the way $f_h$ is chosen, there is a directed path from $f_i$ to $f_p$ in $C$. Therefore, there is a directed path from $f_i$ to $f_h$ in $C$ --- a contradiction. \qed
\end{proof}

\begin{lemma}\label{lem:flip-restore}
Let $f_i$ and $f_h$, where $i < h$, be two flips in $C$ such that $\phi(f_h)=\epsilon(f_i)$, and $\epsilon(f_i)$ is not flipped between $f_i$ and $f_h$. There is a directed path from $f_i$ to $f_h$ in $C$.
\end{lemma}

\begin{proof}
If $\phi(f_i) =\epsilon(f_h)$, then $f_i \rightarrow f_h$ and the statement of the lemma follows. Suppose now that $\phi(f_i) \neq \epsilon(f_h)$. Since $\phi(f_h)=\epsilon(f_i)$, both $\epsilon(f_i)$ and $\epsilon(f_h)$ are diagonals of $Q_{\epsilon(f_h)}$, and hence they cross. Therefore, $\epsilon(f_h)$ is not an edge of ${\cal T}_{i-1}$ (the triangulation just before
flip $f_i$). Since $\epsilon(f_h)$ is an edge in ${\cal T}_{h-1}$ (the triangulation just before flip $f_h$) and $\phi(f_i) \neq \epsilon(f_h)$, there must exist a flip between
$f_i$ and $f_h$ that created $\epsilon(f_h)$. Let $f_p$ be the last flip between $f_i$ and $f_h$ such that $\phi(f_p) = \epsilon(f_h)$, and note that there is an arc from $f_p$ to $f_h$ in $C$. Then $\phi(f_p)$ (which is $\epsilon(f_h)$) crosses $\epsilon(f_i)$, and $\epsilon(f_i)$ is not flipped between $f_i$ and $f_p$ in $C$. Since $\epsilon(f_i)$ is not flipped between $f_i$ and $f_h$, $\epsilon(f_p) \neq \epsilon(f_i)$. By Lemma~\ref{lem:crossing}, there is a directed path from $f_i$ to $f_p$ in $C$. Combining this directed path with the arc from $f_p$ to $f_h$, we obtain a directed path from $f_i$ to $f_h$ in $C$. \qed
\end{proof}

\begin{lemma}\label{lem:flip-flip}
Let $f_i$ and $f_h$, where $i < h$, be two flips in $C$ such that $\epsilon(f_i)=\epsilon(f_h)$. There is a directed path from $f_i$ to $f_h$ in $C$.
\end{lemma}
\begin{proof}
Assume that the statement is not true, and let $f_i$ and $f_h$ be the closest pair of flips in $C$ such that $\epsilon(f_i)=\epsilon(f_h)$ and there is no directed path from $f_i$ to $f_h$ in $C$. If $\epsilon(f_i)=\epsilon(f_p)=\epsilon(f_h)$ for some $f_p$ between $f_i$ and $f_h$, then by the way $f_i$ and $f_h$ are chosen, there is a directed path from $f_i$ to $f_p$ and a directed path from $f_p$ to $f_h$, which implies that there is a directed path from $f_i$ to $f_h$
--- a contradiction.

Now we can assume that $\epsilon(f_i)$ is not flipped between $f_i$ and $f_h$. Let $f_p$ be the last flip before $f_h$ such that $\phi(f_p)=\epsilon(f_h)$; then there is an arc from $f_p$ to $f_h$.
Since $\phi(f_p)=\epsilon(f_h)=\epsilon(f_i)$, by Lemma~\ref{lem:flip-restore}, there is a directed path from $f_i$ to $f_p$ in $C$. Combining this directed path with the arc from $f_p$ to $f_h$, we obtain a directed path from $f_i$ to $f_h$ in $C$ --- a contradiction. \qed
\end{proof}

\begin{lemma}\label{lem:path}
Let $f_i$ and $f_h$, where $i < h$, be two flips in $C$. If $\phi(f_i) =\epsilon(f_h)$, or if $\phi(f_i)$ and $\epsilon(f_h)$ share a triangle $T$ in ${\cal T}_j$, for some $j$ satisfying $i \leq j < h$, then there is a directed path from $f_i$ to $f_h$ in $C$.
\end{lemma}

\begin{proof}
First consider the case when $\phi(f_i) =\epsilon(f_h)$. If $\phi(f_i) =\epsilon(f_h)$ is not flipped between $f_i$ and $f_h$, then there is an arc from $f_i$ to $f_h$ in $C$. Suppose now that $\phi(f_i) =\epsilon(f_h)$ is flipped by another flip between $f_i$ and $f_h$. Let $f_p$ be the first such flip. Then there is an arc from $f_i$ to $f_p$ in $C$, and by Lemma~\ref{lem:flip-flip}, there is a directed path from $f_p$ to $f_h$. Therefore there is a directed path from $f_i$ to $f_h$ in $C$.

Now consider the case where $\phi(f_i)$ and $\epsilon(f_h)$ share a triangle $T$ in ${\cal T}_j$, for some $j$ satisfying $i \leq j < h$, and assume, for the sake of a contradiction, that there is no directed path from $f_i$ to $f_h$ in $C$. Without loss of generality, assume that $f_i$ and $f_h$ are the closest such pair of flips in $C$.

Note that $f_h$ cannot be the flip that immediately follows $f_i$, because in such case there would be an arc from $f_i$ to $f_h$ in $C$.

If $\phi(f_i)$ is flipped by a flip $f_p$ satisfying $i < p \leq j$, then in this case $\phi(f_i)$ must be restored by a flip $f_q$ where $p < q \leq j$ since $\phi(f_i)$ is in ${\cal T}_j$. We can assume that $f_p$ and $f_q$ are the first such flips. Therefore, there is an arc from $f_i$ to $f_p$. Since $\phi(f_q)=\epsilon(f_p)$ and $\epsilon(f_p)$ is not flipped between $f_p$ and $f_q$, by Lemma~\ref{lem:flip-restore}, there is a directed path from $f_p$ to $f_q$ in $C$. Since $\phi(f_q)$ (which is $\phi(f_i)$) and $\epsilon(f_h)$ share a triangle in ${\cal T}_j$, $q \leq j < h$, and the number of flips between $f_q$ and $f_h$ is less than that between $f_i$ and $f_h$, by the way $f_i$ and $f_h$ are chosen, there is a path from $f_q$ to $f_h$ in $C$. Therefore, there is a directed path from $f_i$ to $f_h$ in $C$ --- a contradiction.

Now we can assume that $\phi(f_i)$ is not flipped by a flip $f_p$ satisfying $i < p \leq j$. If no edge of the triangle $T$ is flipped by a flip $f_p$ between $f_j$ and $f_h$, then there is an arc from $f_i$ to $f_h$ in $C$--- a contradiction.

Finally, consider the case when an edge of $T$ is flipped by $f_p$ for the first time, where $j < p < h$. This means that $\phi(f_i)$ is in ${\cal T}_{p-1}$, and hence there is an arc from $f_i$ to $f_p$.
If $\epsilon(f_p) \neq \epsilon(f_h)$, then $\epsilon(f_p)$ and $\epsilon(f_h)$ share a triangle in ${\cal T}_{p-1}$, which means $\phi(f_p)$ and $\epsilon(f_h)$ share a triangle in ${\cal T}_p$.
Since the number of flips between $f_p$ and $f_h$ is less than that between $f_i$ and $f_h$, by the way $f_i$ and $f_h$ are chosen, there is a path from $f_p$ to $f_h$ in $C$. Combining this path with the arc from $f_i$ to $f_p$ gives a directed path from $f_i$ to $f_h$ in $C$--- a contradiction. On the other hand, if $\epsilon(f_p) = \epsilon(f_h)$, then by Lemma~\ref{lem:flip-flip}, there is a directed path from $f_p$ to $f_h$. Combining this path with the arc from $f_i$ to $f_p$ gives a directed path from $f_i$ to $f_h$ in $C$ --- a contradiction. \qed \end{proof}

Lemmas~\ref{lem:crossing}--\ref{lem:path} can be summarized and strengthened by the following lemma:

\begin{lemma}\label{lem:allpath}
Let $f_i$ and $f_h$, where $i < h$, be two flips in $C$. If one of the following conditions is true, then there is a directed path from $f_i$ to $f_h$ in $C$:
\begin{itemize}
\item[(1)] $\phi(f_h)$ crosses $\epsilon(f_i)$.
\item[(2)] $\phi(f_h)=\epsilon(f_i)$.
\item[(3)] $\epsilon(f_i)=\epsilon(f_h)$.
\item[(4)] $\phi(f_i) =\epsilon(f_h)$, or $\phi(f_i)$ and $\epsilon(f_h)$ share a triangle $T$ in ${\cal T}_j$, for some $j$ satisfying $i \leq j < h$.
\end{itemize}
\end{lemma}
\begin{proof}
This lemma follows from Lemmas~\ref{lem:crossing},~\ref{lem:flip-restore},~\ref{lem:flip-flip},~\ref{lem:path}. Note that the condition in the statements of Lemmas~\ref{lem:crossing},~\ref{lem:flip-restore} that $\epsilon(f_i)$ is not flipped between $f_i$ and $f_h$ can be removed because of Lemma~\ref{lem:flip-flip}. \qed
\end{proof}

\section{The algorithm}\label{sec:algorithm}
In Section~\ref{sec:structural}, we showed that we can represent a normalized solution $F$ to an instance ${\cal I}$ of {\sc Flip Distance} using a DAG ${\cal D}_F$, such that any topological sorting of ${\cal D}_F$
is a valid solution to ${\cal I}$.  (Recall that, by definition, a normalized solution is minimum.) Therefore, to solve ${\cal I}$, it suffices to compute a topological sorting of ${\cal D}_F$.
There are two major difficulties that need to be overcome before we can compute a topological sorting of ${\cal D}_F$ within the desired time upper bound.

First, given an instance ${\cal I}$ of {\sc Flip Distance}, we do not know the DAG ${\cal D}_F$ of a normalized solution $F$ to ${\cal I}$ (in case it exists). Therefore, ${\cal D}_F$ can only serve as a representation of $F$ to help us reason about $F$ and eventually construct it. Second, there is a difference between what goes on in the DAG ${\cal D}_F$ and what goes on in a current triangulation ${\cal T}_{current}$, in the sense that the flips in ${\cal D}_F$ do not all correspond to existing edges in ${\cal T}_{current}$, and in fact they usually do not. For example, if $f_i$ and $f_j$ are two adjacent flips in ${\cal D}_F$ ($f_i \rightarrow f_j$), and if the underlying edge $\epsilon(f_j)$ of $f_j$ is an edge in ${\cal T}_{current}$, the edge $\epsilon(f_i)$ may not be present in ${\cal T}_{current}$. Therefore, assuming that we know a certain edge $\epsilon(f_j)$ is present in ${\cal T}_{current}$ and that it needs to be flipped (\eg, $\epsilon(f_j)$ is a changed edge), it is not clear which edge in ${\cal T}_{current}$ should be flipped first.

There is a simple $\FPT$ algorithm for {\sc Flip Distance} that runs in $\Ostar(c^{k^2})$ time, for some constant $c$, and that is much simpler to obtain using the structural results we derive in this paper.
Describing this algorithm will perhaps give more insight into the problem and the solution that we ultimately present. This (simpler) algorithm starts from an edge in the current triangulation that corresponds to a flip in the DAG representing the remaining solution (initially this edge is a changed edge), and grows a BFS-like tree of size $c^k$ searching for the next edge in the current triangulation that can be correctly flipped (corresponding to a source node in the DAG); the algorithm then flips this edge. (Lemma~\ref{lem:walk} guarantees that the next edge can be found using a search tree of depth $k$ and constant outbranching, and hence of size $c^k$, for some constant $c$.) Repeating this process $k$ times gives an $\Ostar((c^{k})^k) = \Ostar(c^{k^2})$-time algorithm for the problem.

Our goal, however, is to obtain an $\Ostar(c^{k})$-time algorithm for the problem, for some constant $c$. Achieving this goal turns out to be quite involved.  We did so by analyzing the relation between the DAG associated with a solution to a problem instance and the changing structure of the underlying triangulations.

\subsection{Overview of the algorithm}
In this subsection, we give an intuitive description of how the algorithm works. Let $({\cal T}_{initial}, {\cal T}_{final}, k)$ be an instance of
{\sc Flip Distance}. In order to solve the instance, by the results and discussions in Section~\ref{sec:structural}, it suffices for the algorithm to perform a sequence of at most $k$ flips that is a topological sorting of the DAG ${\cal D}_F$ associated with a normalized solution $F$ to the instance. The algorithm is a search-tree algorithm that searches for such a sequence of flips; we explain next how the algorithm does that.

Each node $\alpha$ in the search tree $\Upsilon$ of the algorithm is associated with a triplet $({\cal T}_{current}, e_{current}, \Lambda)$, where ${\cal T}_{current}$ and $e_{current} \in {\cal T}_{current}$ are the current triangulation and edge that the algorithm is processing at $\alpha$, respectively, and $\Lambda$ is a stack used by the algorithm as an auxiliary structure. At a node $\alpha$ of $\Upsilon$ associated with triplet $({\cal T}_{current}, e_{current}, \Lambda)$, the algorithm branches into constantly-many branches, associated in a one-to-one fashion with the constantly-many children of $\alpha$. Each branch corresponds to an {\em action} performed by the algorithm, chosen from within a set of constantly-many actions (given in Subsection~\ref{subsec:actions}). An action can be a {\em flip-action}, in which edge $e_{current}$ is flipped to result in a triangulation, and an edge in the resulting triangulation is chosen to be processed next along that branch; or it can be a {\em move-action}, which results in the algorithm processing next along that branch an edge in ${\cal T}_{current}$ that is adjacent to $e_{current}$. Both action categories may possibly consult the stack $\Lambda$ and/or modify it by popping or pushing an edge. An action results in a new triplet that is associated with the child of $\alpha$ corresponding to the branch/action. The root-node of $\Upsilon$, where the algorithm starts its search, is associated with the triplet $({\cal T}_{initial}, e, \Lambda=\emptyset)$, where $e$ is any changed edge in ${\cal T}_{initial}$.

We will prove in Subsection~\ref{sec:component} that $({\cal T}_{initial}, {\cal T}_{final}, k)$ is a yes-instance of {\sc Flip Distance} if and only if there is a root-leaf path in $\Upsilon$ of length $O(k)$, on which the number of flip-actions is at most $k$, and such that the triangulation at the leaf of this path (\ie, in the triplet at the leaf) is ${\cal T}_{final}$. Since each node in $\Upsilon$ has constantly-many children, this will yield a search-tree algorithm of running time $\Ostar(c^{k})$.

\subsection{The algorithm and its actions}\label{subsec:actions}
  The root-node of the search-tree $\Upsilon$ of the algorithm is associated with the triplet $({\cal T}_{initial}, e, \Lambda=\emptyset)$, where $e$ is any changed edge in ${\cal T}_{initial}$.  At a node $\alpha$ in $\Upsilon$, associated with triplet $({\cal T}_{current}, e, \Lambda)$, the algorithm branches into one or more branches, each corresponding to performing an action $\sigma$ by the algorithm. Each of these branches corresponds to a child $\alpha'$ of $\alpha$ in $\Upsilon$, and will be associated with a triplet $({\cal T'}, e', \Lambda')$. Note that a triangulation ${\cal T}_{current}$ at $\alpha$ is the result of applying the actions along the path in $\Upsilon$ from the root to $\alpha$.  The actions/branches performed at a (non-leaf) node $\alpha$ in $\Upsilon$ associated with triplet $({\cal T}_{current}, e, \Lambda)$ are the following, which we classify into five types\footnote{If an action $\sigma$ involves flipping an edge $e$, then $\sigma$ is performed only when a flip to $e$ is admissible.}:

\begin{itemize}
\item[(i)]  For each edge $e'$ that shares a triangle with $e$ in ${\cal T}_{current}$: branch into a node $\alpha'$ in $\Upsilon$ associated with the triplet $({\cal T'} ={\cal T}_{current}, e', \Lambda'=\Lambda)$.
\item[(ii)] For each edge $e'$ that shares a triangle with $e$ in ${\cal T}_{current}$: branch into a node $\alpha'$ of $\Upsilon$ associated with the triplet $({\cal T'}, e', \Lambda'=\Lambda)$, where ${\cal T'}$ is the triangulation obtained from ${\cal T}_{current}$ by flipping $e$.
\item[(iii)]  For each edge $e'$ that shares a triangle with $e$ in ${\cal T}_{current}$: branch into a node $\alpha'$ of $\Upsilon$ associated with the triplet $({\cal T'}, e', \Lambda')$, where ${\cal T'}$ is the triangulation obtained from
            ${\cal T}_{current}$ by flipping $e$, and $\Lambda'$ is obtained from $\Lambda$ by pushing the edge $\phi(e)$ (on top of $\Lambda$).

\item[(iv)]  Branch into a node $\alpha'$ of $\Upsilon$ associated with the triplet $({\cal T'}, e', \Lambda)$, where ${\cal T'}$ is the triangulation obtained from
            ${\cal T}_{current}$ by flipping $e$, and $e'$ is the edge on the top of $\Lambda$.

\item[(v)] Branch into a node $\alpha'$ of $\Upsilon$ associated with the triplet $({\cal T'}, e', \Lambda')$, where ${\cal T'}$ is the triangulation obtained from
            ${\cal T}_{current}$ by flipping $e$, $e'$ is the edge on the top of $\Lambda$, and $\Lambda'$ is obtained by popping $e'$ from $\Lambda$.
\end{itemize}

\begin{proposition}\label{prop:choices}
The total number of actions/branches at any node in the search tree is at most 14.
\end{proposition}
\begin{proof}
At any node in the search tree, there are at most 4 actions/branches for each of types (i)-(iii), and 1 action for each of types (iv)-(v). \qed
\end{proof}

We are now ready to present the algorithm.  Let $({\cal T}_{initial}, {\cal T}_{final}, k)$ be an instance of {\sc Flip Distance}.  The algorithm searches for a sequence of at most $k$ flips that corresponds to a topological sorting of the DAG ${\cal D}_F$, associated with a normalized solution $F$ to the instance (in case such a solution exists). We will show in the next section that if $F$ exists, then there is a sequence of actions of length at most $11k$, in which the flip-actions correspond to a topological sorting of ${\cal D}_F$. Therefore, all that the algorithm needs to do is search the space of all possible sequences of actions of length at most $11k$ for such a sequence. If such a sequence exists, the algorithm accepts. Otherwise, the algorithm rejects.

To implement the above, the algorithm starts by ordering the changed edges in ${\cal T}_{initial}$ arbitrarily; denote this ordering by $\Gamma$. The algorithm then tries (enumerates) each positive integer $t$, such that $t \leq k$, as the correct number of components in ${\cal D}_F$. For a fixed $t$, the algorithm tries (enumerates) each tuple of integers $(k_1, \ldots, k_t)$ satisfying $k_1, \ldots, k_t \geq 1$ and $k_1+k_2+\cdots +k_t \leq k$, as the tuple in which $k_{\ell}$ corresponds to the correct number of flips in component $C_{\ell}$ of ${\cal D}_F$, for $\ell=1, \ldots, t$. For a fixed $t$ and a fixed tuple $(k_1, \ldots, k_t)$, the algorithm constructs a search tree $\Upsilon$ over $t$ branching stages, $\ell=1, \ldots, t$, as follows, where in stage $\ell$ the algorithm searches for a sequence of flips that corresponds to a topological sorting of component $C_{\ell}$. Let $\alpha$ be the node of $\Upsilon$ that the algorithm is at during its search after the first $\ell-1$, $\ell=1, \ldots, t$, stages have been completed; initially $\alpha$ is the root of $\Upsilon$. Let ${\cal T}$ be the triangulation at $\alpha$ (\ie, in the triplet associated with $\alpha$), which is the result of applying all the actions along the path from the root of $\Upsilon$ to $\alpha$; initially, at the root of $\Upsilon$, ${\cal T} = {\cal T}_{initial}$. At the beginning of stage $\ell$, at node $\alpha$, the algorithm picks the first changed edge $e \in \Gamma$ with respect to ${\cal T}$ and ${\cal T}_{final}$ (\ie, $e$ is in ${\cal T}$ but not in ${\cal T}_{final}$). The algorithm associates with $\alpha$ the triplet $({\cal T}, e, \Lambda=\emptyset)$, and then branches into all possible sequences ${\cal S}$ of actions (of types (i)-(v) described before) of length at most $11k_{\ell}$ in which the number of flip-actions is at most $k_{\ell}$. After the last stage $\ell=t$ has been completed, if at a leaf-node of $\Upsilon$ we reached triangulation ${\cal T}_{final}$, then the algorithms accepts the instance and stops. If no branch in the search tree accepts, the algorithms rejects the instance. The algorithm, referred to as \algo, is described schematically in {\bf Algorithm~1} below.

\renewcommand{\figurename}{Algorithm}
\setcounter{figure}{0}

\begin{algorithmnew}{The algorithm for {\sc Flip Distance}.}{alg:wholealgo}
\algtitle{\algo}
\begin{codeblock}
\step Order the changed edges in ${\cal T}_{initial}$ arbitrarily and denote this ordering by $\Gamma$.

\step For each $t=1, \ldots, k$, and for each tuple $(k_1, \ldots, k_t)$ satisfying $k_1+k_2+\cdots +k_t \leq k$ do:

\begin{itemize}
\step[2.1.]  Start a new search tree $\Upsilon$, where ${\cal T}_{initial}$ is the triangulation at the root.
\step[2.2.]  Perform $t$ branching stages as follows, where the branching in stage $\ell$, for $\ell =1$ to $t$, is performed starting from {\em each} leaf-node $\alpha$ of the search tree resulting at the end of stage $\ell-1$ (if $\ell = 1$ then there is only one leaf-node $\alpha$, which is the root of the search tree):
\begin{itemize}
\item[2.2.1.] Let ${\cal T}$ be the triangulation at node $\alpha$ in $\Upsilon$.
\item[2.2.2.] Let $e$ be the first edge in $\Gamma$ with respect to the ordering (if $\ell =1$ then $e$ is the first edge in $\Gamma$) such that $e$ is in ${\cal T}$ but not in ${\cal T}_{final}$. Do the following:
 \begin{itemize}
 \item[2.2.2.1.] Associate the triplet $({\cal T}, e, \Lambda=\emptyset)$ with node $\alpha$ of the search tree.
 \item[2.2.2.2.] Branch into all possible sequences ${\cal S}$ of actions of length at most $11k_{\ell}$, in which the number of flips is at most $k_{\ell}$. For each such sequence ${\cal S}$ do the following:
\begin{itemize}
\item[2.2.2.2.1.] Let ${\cal T'}$ be the resulting triangulation after performing ${\cal S}$ starting at $\alpha$.
\item[2.2.2.2.2.] If $\ell =t$ and ${\cal T'} = {\cal T}_{final}$ then accept the instance and exit the (whole) algorithm.
\end{itemize}
\end{itemize}
\end{itemize}
\end{itemize}
\step Reject the instance.
\end{codeblock}
\end{algorithmnew}

\renewcommand{\figurename}{Fig.}
\setcounter{figure}{1}

The next section is dedicated to proving the correctness of \\ \algo~and analyzing its running time.

\subsection{The correctness and analysis of \algo}\label{sec:component}
 Let $F=\langle f_1, \ldots, f_r \rangle$, where $r \leq k$, be a normalized solution to an instance $({\cal T}_{initial}, {\cal T}_{final}, k)$ of {\sc Flip Distance}. Let $C$ be a component of ${\cal D}_F$. By Corollary~\ref{cor:ts}, we can assume that all the flips in $C$ appear at the beginning of $F$, that is, form a prefix of $F$; let $F_C = \langle f_1, \ldots, f_t \rangle$ be the prefix of $F$ corresponding to the flips in $C$. The correctness of \algo~and the claimed upper bound on its running time will follow easily from the following key theorem, which most of this section is dedicated to proving:

\begin{theorem}\label{thm:traverse}
Let $C$ be a component of ${\cal D}_F$. There is a sequence of actions of length at most $11|V(C)|$ that, starting from any changed edge $\epsilon(f_h)$, where $f_h \in C$, performs the flips in $C$ in a topologically-sorted order.
\end{theorem}

Clearly, if we can show the existence of the sequence in Theorem~\ref{thm:traverse}, then there is a branch (in step 2.2.2.2) in the search tree depicted by \\ \algo~that will find this sequence. Since there are 14 possible actions at each node in the search tree, this implies that the number of branches in the search tree that \algo~branches into in order to find the desired sequence is $\Oh(14^{11|V(C)|})$, and in order to find the desired sequence for the whole DAG ${\cal D}_F$ is
$\Oh(14^{11|V({\cal D}_F)|})=\Oh(14^{11k})$.

To prove Theorem~\ref{thm:traverse}, we define a spanning subgraph $J_C$ of the underlying graph of $C$ recursively. We then exhibit a sequence ${\cal S}$ of actions for \algo~to perform that can be depicted by a recursive traversal of $J_C$. By that we mean that the actions performed by the algorithm (\ie, the actions in ${\cal S}$) in the triangulations correspond to a traversal of the edges and nodes of $J_C$, and such that the sequence of flips performed by the algorithm (\ie, in ${\cal S}$) is a topological sorting of $C$.  We initialize $J_C$ to be empty, and we start the recursive definition of $J_C$ at a node in $C$ that corresponds to any changed edge in the current triangulation. We will then add edges and nodes to $J_C$, and recurse on the connected components of the graph resulting from $C$ after a source node in $C$ has been removed. Since during the recursion nodes and edges get removed from $C$, the resulting graph of $C$ may consist of several connected components that we will refer to as ``chunks'' (formally defined in Definition~\ref{def:chunk}), in order to distinguish them from the components of ${\cal D}_F$. Assume that the current triangulation is ${\cal T}$ when we recurse on a chunk $H$ to define its spanning subgraph $J_H$. The recursive call starts at a node $f_h$ in $H$ that we call the ``entry point'' of $H$. At the top level of the recursion, $C$ is the only chunk (in the recursive definition), and the entry point of $H=C$ is any node in $C$ corresponding to any changed edge in the current triangulation. We will define in Lemma~\ref{lem:walk} a directed path $B=\langle b_1=f_s, \ldots, b_{\ell}=f_h\rangle$ in $H$ from a source node $f_s$ in $H$ to the entry point $f_h$ of $H$. With the path $B$, we correspond a {\em walk} $W$, defined in Lemma~\ref{lem:walk}, in the current triangulation from $\epsilon(f_h)$ to $\epsilon(f_s)$. We add $B$ to $J_H$, we add the edges between $f_s$ and the entry point of each chunk in $H-f_s$ to $J_H$, and we recurse on the chunks of $H-f_s$ to complete the recursive definition of $J_H$. The corresponding actions in ${\cal S}$ (with the recursive definition of $J_H$) consist of those for performing the walk $W$, flipping $\epsilon(f_s)$, and recursively performing the sequence of actions corresponding to the traversals of the chunks in $H-f_s$.  Note that to flip a single edge, the algorithm needs to perform a sequence of actions that corresponds to a walk (in the current triangulation) to a source node in $H$. Therefore, if we are not careful in how we do the traversal of $C$, the length of all these walks could be quadratic in $k$.  To ensure that when the algorithm is done performing the sequence of actions in a chunk it can continue with the other chunks, the algorithm uses the stack $\Lambda$ to store the edge $\phi(f_s)$, resulting from flipping the ``connecting point'' $f_s$ of all these chunks, so that the algorithm, after performing all the flips in a chunk of $H-f_s$, can go back by a single action to $\phi(f_s)$. We formalize next the concepts discussed above, starting with the following definitions:

\begin{definition}\label{def:walk}
A {\em walk} in a triangulation ${\cal T}$ (starting from an edge $e \in {\cal T}$) is a sequence of edges of ${\cal T}$ (starting from $e$) such that any two consecutive edges in this sequence share a triangle in ${\cal T}$. Note that a walk can be performed by a sequence of actions all of which are of type (i).
\end{definition}

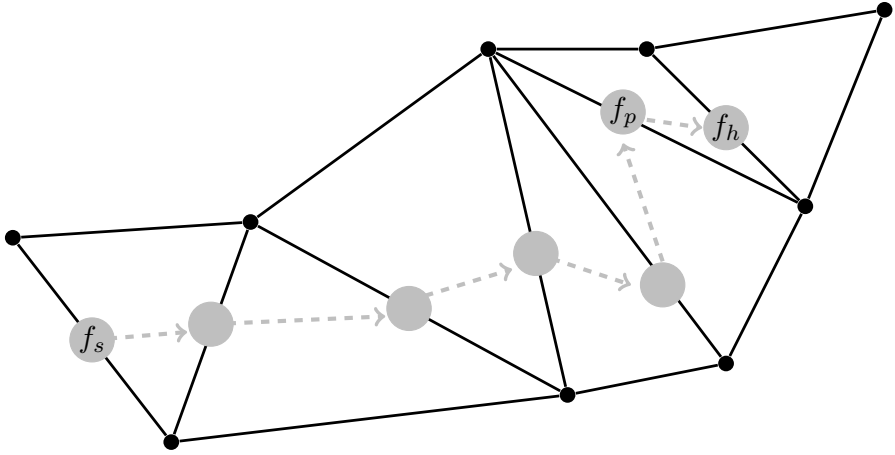
\begin{figure}
\begin{center}
 \resizebox{1\textwidth}{!}{
\begin{tikzpicture}
[vertex/.style={fill,circle,inner sep = 2pt}]

\node (a) at (9, 2) [vertex]{};

\node (d) at (13,0) [vertex] {};
\node (b) at (11,2) [vertex]{};
\node (c) at (14,2.5) [vertex] {};
\node (e) at (12,-2) [vertex] {};
\node (f) at (5,-3) [vertex] {};
\node (g) at (3,-0.4) [vertex] {};

\node (x) at (6,-0.2) [vertex] {};
\node (y) at (10,-2.4) [vertex] {};

\draw[line width=1] (a) -- (b) -- (c) -- (d);
\draw[line width=1] (a) -- (d);
\draw[line width=1] (d) -- (e) -- (a);
\draw[line width=1] (f) -- (y) -- (e);
\draw[line width=1] (a) -- (x) -- (g) -- (f);
\draw[line width=1] (x) -- (y);
\draw[line width=1] (x) -- (f);
\draw[line width=1] (a) -- (y);

\draw[black, line width=1] (b) -- (d);

\filldraw[fill=lightgray, draw=lightgray] (12, 1) circle[radius= 0.28cm] node {\large $f_h$};
\filldraw[fill=lightgray, draw=lightgray] (10.7, 1.2) circle[radius= 0.28cm] node {\large $f_p$};
\filldraw[fill=lightgray, draw=lightgray] (11.2, -1) circle[radius= 0.28cm] ;
\filldraw[fill=lightgray, draw=lightgray] (4, -1.7) circle[radius= 0.28cm] node {\large $f_s$};
\filldraw[fill=lightgray, draw=lightgray] (5.5, -1.5) circle[radius= 0.28cm] node {};
\filldraw[fill=lightgray, draw=lightgray] (8, -1.3) circle[radius= 0.28cm] node {};
\filldraw[fill=lightgray, draw=lightgray] (9.6, -0.6) circle[radius= 0.28cm] node {};



\draw[->,lightgray, line width=1.5,dashed] (11, 1.1)  -- (11.7, 1);
\draw[->,lightgray, line width=1.5,dashed] (11.2, -0.7) -- (10.7, 0.8);
\draw[->,lightgray, line width=1.5,dashed] (4, -1.7) -- (5.2, -1.6);
\draw[->,lightgray, line width=1.5,dashed]  (5.5, -1.5) -- (7.7, -1.4);
\draw[->,lightgray, line width=1.5,dashed] (8, -1.2) -- (9.3, -0.8);
\draw[->,lightgray, line width=1.5,dashed] (9.6, -0.6) -- (10.8, -1);
\end{tikzpicture}
}
\end{center}
\caption{Illustration for the proof of Lemma~\ref{lem:walk}. The black vertices and edges are vertices and edges of ${\cal T}$, respectively. The gray circles represent the nodes/flips in $C$, the black edges (through them) are their underlying edges in ${\cal T}$, and the dashed edges are arcs in $C$. }
\label{fig:walk}
\end{figure}

\begin{definition}\rm \label{def:chunk}
A subgraph $H$ of $C$ is called a \emph{chunk} of $C$ if either $H=C$, or if there exists a sequence $s_1, \ldots, s_{j}$ of nodes of $C$ such that: (1) $H$ is a component of $C -\{s_1, \ldots, s_j\}$; and (2) $s_1$ is a source node of $C$ and $s_i$ is a source node of $C-\{s_1, \ldots, s_{i-1}\}$, for $i=2,\ldots, j$. That is, $H$ is a chunk if $H$ is $C$ itself, or if it is a component of a subgraph obtained from $C$ by ``repeatedly removing source nodes,'' starting with a source node of $C$, and continuing each time with a source node of the resulting subgraph of $C$. (Note that $s_1$, the first node in such a sequence, is a source node of $C$, but the remaining nodes $s_2, \ldots, s_j$ in the sequence  may not be source nodes of $C$.)
\end{definition}

\begin{lemma}\rm\label{lem:walk}
Suppose that $H$ is a chunk obtained from $C$ by removing a (possibly empty) sequence of nodes $F^-$. (Note that each node in $F^-$ is a source node in the DAG resulting from $C$ after removing the preceding nodes in $F^-$.) Let $f_h$ be a node in $H$ such that $\epsilon(f_h)$ is an edge in the triangulation ${\cal T}$ resulting after performing the sequence of flips $F^-$. There is a source node $f_s$ in $H$ satisfying:
(1) there is a walk $W$ in ${\cal T}$ from $\epsilon(f_h)$ to $\epsilon(f_s)$;
(2) there is a directed path $B$ from $f_s$ to $f_h$ in $H$ that we refer to as the {\em backbone} of $H$ (note that $B$ is not necessarily unique); and
(3) the length of $W$ is at most that of $B$.
\end{lemma}

\begin{proof}
For any node $f_i$ in $H$, define the {\em height} of $f_i$ to be the length of a {\em longest} path from {\em any} source node of $H$ to $f_i$. The proof is by induction on the height of $f_h$ in $H$. We refer the reader to Fig.~\ref{fig:walk} for illustration.

If the height of $f_h$ is 0 then $f_h$ is a source node in $H$, the path $B$ consists of $f_h$, and the walk $W$ has length 0. Therefore, the statement is trivially true for the base case.

Now assume that the height of $f_h$ is not zero, and hence $f_h$ is not a source node in $H$. Since $\epsilon(f_h)$ is an edge in ${\cal T}$, let $Q_{\epsilon(f_h)}$ be the quadrilateral associated with $\epsilon(f_h)$ in ${\cal T}$. We first argue that there must exist a flip in $H$ to a boundary edge of $Q_{\epsilon(f_h)}$. Since $H$ is a chunk of $C$ obtained by removing the sequence of nodes $F^-$, the sequence $F^-$ is a prefix of a topological sorting of ${\cal D}_F$, and hence of $C$. Let $\pi(F)$ be {\em any} topological sorting of ${\cal D}_F$ of which $F^-$ is a prefix. By Lemma~\ref{lem:ts}, the DAGs ${\cal D}_F$ and ${\cal D}_{\pi(F)}$ are the same. Since $f_h$ is not a source node in $H$, there exists a node $f_q$ in $H$ such that $f_h$ is adjacent to $f_q$ in $H$. Since $H$ is a subgraph of $C$ (and hence of ${\cal D}_F$) and ${\cal D}_{\pi(F)} = {\cal D}_F $, this implies that $f_h$ is adjacent to $f_q$ in ${\cal D}_{\pi(F)}$ as well. By the definition of adjacency in ${\cal D}_{\pi(F)}$, either $\phi(f_q)=\epsilon(f_h)$, or $\epsilon(f_h)$ and $\phi(f_q)$ share a triangle in the triangulation preceding flip $f_h$, when the flips are performed (starting from ${\cal T}_{initial}$) with respect to the sequence $\pi(F)$.
 Since $f_h$ is adjacent to $f_q$ in $H$, if no boundary edge of $Q_{\epsilon(f_h)}$ is flipped before flip $f_h$ in $\pi(F)-F^-$, then this is only possible if $\epsilon(f_h)$ is flipped by some flip $f_x$ that appears before $f_q$ in $\pi(F)-F^-$, and then restored back by flip $f_q$ (\ie, $\phi(f_q) = \epsilon(f_h))$; this is, however, a contradiction to the minimality of $\pi(F)$ (because $F$ is minimum and $|\pi(F)|=|F|$).

Therefore, one of the boundary edges of $Q_{\epsilon(f_h)}$ must be flipped by some flip before $f_h$ in $\pi(F) - F^-$; let $f_p$ be the first such flip in $\pi(F) - F^-$. Since $\epsilon(f_p)$ and $\epsilon(f_h)$ share a triangle before $f_p$ is flipped in $\pi(F) - F^-$, $\phi(f_p)$ and $\epsilon(f_h)$ share a triangle after $f_p$ is flipped in $\pi(F) - F^-$, and by Lemma~\ref{lem:allpath}, there is a directed path from $f_p$ to $f_h$ in the component $C$ of ${\cal D}_{\pi(F)}={\cal D}_{F}$. Since $H$ was obtained from $C$ by repeatedly removing source nodes (if any), there is a directed path from $f_p$ to $f_h$ in $H$. The edges $\epsilon(f_h)$ and $\epsilon(f_p)$ share a triangle in ${\cal T}$, and hence, in one action (of type (i)) the algorithm can move from $\epsilon(f_p)$ to $\epsilon(f_h)$ in ${\cal T}$.

Since there is a directed path from $f_p$ to $f_h$ in $H$, the height of $f_p$ is smaller than that of $f_h$. By the inductive hypothesis, there is a source node $f_s$ in $H$, such that (1) there is  a walk $W'$ from $\epsilon(f_p)$ to $\epsilon(f_s)$ in ${\cal T}$; (2) there is a directed path $B'$ from $f_s$ to $f_p$ in $H$; and (3) the length of $W'$ is at most that of $B'$.

From the above, it follows that (1) there is a walk $W$ from $\epsilon(f_h)$ to $\epsilon(f_s)$ in ${\cal T}$ by extending $W'$ by a single move from $\epsilon(f_h)$ to $\epsilon(f_p)$; (2) there is a path $B$ from $f_s$ to $f_h$ in $H$ by concatenating $B'$ with the directed path from $f_p$ to $f_h$ in $H$; and (3) the length of $W$ is at most that of $B$. This concludes the inductive proof. \qed
\end{proof}

We now formally give the recursive definition of $J_C$, described for a chunk $H$ of $C$ obtained from $C$ during the recursion. The definition uses a designated node $f_h$ in $H$, referred to as the ``entry point'' of $H$. We will associate with each chunk $H$, obtained during the recursive construction, a number $\mu(H)$, referred to as the {\em level} of $H$, indicating the recursive phase-number during which the chunk was first created. At the top level of the recursion, $H=C$, $f_h$ is a node in $C$ corresponding to any changed edge in the initial triangulation, and $\mu(H)=0$; otherwise, the entry point $f_h$ and $\mu(H)$ will be defined recursively.

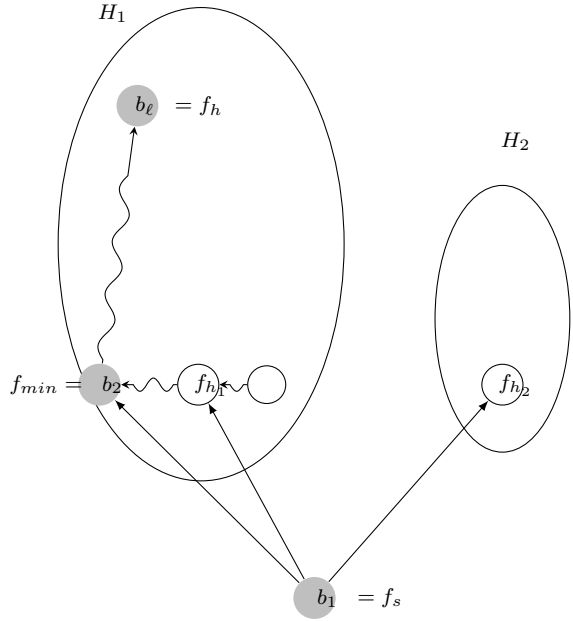
\begin{figure}
\begin{center}
\begin{tikzpicture}[
  every node/.style={on grid},
  setA/.style={fill=lightgray,circle,inner sep=2pt,minimum size=.4cm},
  setC/.style={circle, draw, inner sep=.1pt,minimum size=.5cm},
  every fit/.style={draw,ellipse,text width=70pt},
  >=latex
]

\hspace*{1mm} \node[setA,label=right:{$=f_h$}] (a) {$b_{\ell}$}; \node[left =5mm of a] (w) {};

\hspace*{0.7mm} \node [setA,below=3.7cm of w,label=left:{$f_{min}=$}]
(c) {$b_2$};  \node [above=of w,anchor=south] {$H_1$}; \node [setC, right =1.3cm of c] (z) {$f_{h_1}$}; \node [setC, right =2.2cm of c]
(z1) {}; \node [setC, right =4cm of z] (b) {$f_{h_2}$}; \node [setA,below right =4cm of c,label=right:{$=f_s$}] (e) {$b_1$}; \node[inner sep=0pt,above=2cm of b] (x) {};

\node[above=of x,anchor=south] {$H_2$};

\draw[<-,snake=coil,line before snake=.5mm,line after snake=.5mm,segment aspect=0,
        segment length=20pt,-stealth] (c) -- node[] {} (a); \draw[<-,snake=coil,line before snake=.5mm,line after snake=.5mm,segment aspect=0,
        segment length=10pt,-stealth] (z) -- node[] {} (c);

\draw[<-,snake=coil,line before snake=.5mm,line after snake=.5mm,segment aspect=0,
        segment length=10pt,-stealth] (z1) -- node[] {} (z);

\draw[->] (e) -- node[] {} (c);
\draw[->] (e) -- node[] {} (b);
\draw[->] (e) -- node[] {} (z);

\begin{pgfonlayer}{background}
\node[fit= (a) (z1) ] {};

\tikzset{every fit/.style={draw,ellipse,text width=30pt}} \node[fit=
(b) (x)] {}; \end{pgfonlayer} \end{tikzpicture} \end{center} \caption{Illustration of the definition of entry points, where backbone nodes are colored gray. The entry point of $H_1$ is the node $f_{h_1}$ in $H_1$ adjacent to $f_s$ that has a path to the backbone node $f_{min}$ with the minimum index (node $b_2$ in this case). The entry point of $H_2$, which does not contain backbone nodes in this case, is the node $f_{h_2}$ of minimum index ($h_2$) that is adjacent to $f_s$.} \label{fig:connecting} \end{figure}

\begin{definition}\rm
\label{def:tree}
 Let $H$ be a chunk of $C$, $\mu(H)$ the level of $H$, and $f_h$ the entry point of $H$. (At the top level of the recursion $H=C$, $\mu(H)=0$, and $f_h$ is a node in $C$ corresponding to any changed edge.) The subgraph $J_H$ of $H$ is defined recursively as follows.

\begin{itemize}

\item[1.] Let $B=\langle f_s=b_1, \ldots, b_{\ell}=f_h\rangle$, where $f_s$ is a source node in $H$ (possibly $f_h=f_s$), be the backbone of $H$ defined in Lemma~\ref{lem:walk}. Define the {\em level} of $B$ to be that of chunk $H$ (\ie, $\mu(H)$).

\item[2.]  Remove $f_s$ from $H$, and let $H_1, \ldots, H_x$ be the components of $H^-=H-f_s$. Define $f_s$ to be the {\em connecting point} to chunk $H_p$, and define $\mu(H_p)= \mu(H)+1$, for $p=1, \ldots, x$.

\item[3.] For each chunk $H_p$, $p =1, \ldots, x$, if $H_p$ contains a node from a backbone whose level is smaller than $\mu(H_p)$, then let $f_{min}$ be the backbone node in $H_p$ of minimum index (with respect to $F$); define the {\em entry point} of $H_p$ to be the node $f_{h_p}$ in $H_p$ that is adjacent to $f_s$ and that has a path to $f_{min}$ in $H_p$, and in case more than one neighbor of $f_s$ satisfies this property pick the one with the minimum index (Lemma~\ref{lem:properties} proves the existence of $f_{h_p}$). Otherwise ($H_p$ does not contain a backbone node), define the entry point of $H_p$ to be the node $f_{h_p}$ in $H_p$ with the minimum index $h_p$ (with respect to $F$) that is adjacent to $f_s$. (See Fig.~\ref{fig:connecting} for illustration.)

\item[4.] Define the subgraph $J_p$ of $H_p$ with entry point $f_{h_p}$ recursively, for $p =1, \ldots, x$.

\item[5.] Define $J_H$ to be the union of the edges in $B$, the edges in $J_p$ and the edges between $f_s$ and each entry point of $H_p$, for $p =1, \ldots, x$.
\end{itemize}
\end{definition}

Let $J_C$ be the subgraph of the underlying graph of $C$ resulting from applying the above recursive definition to $C$ starting at any flip in $C$ corresponding to any changed edge. We have the following lemma:

\begin{lemma}
\label{lem:properties}
Let $i \in \nat$. All the backbones --- defined during the recursive definition of $J_C$ --- whose level is at most $i$ and that exist in the same chunk are edge-disjoint and belong to a single (simple) path in the chunk; on this path the (remaining) nodes from each backbone appear consecutively. Moreover, the entry node of a chunk (defined in step 3 of Definition~\ref{def:tree}) exists.
\end{lemma}

\begin{proof}
The proof is by induction on $i$. The statement is clearly true for $i=0$ --- at the top level of the recursion, where the only chunk is $C$ whose entry point is defined to be a flip corresponding to a changed edge, and where there is only one defined backbone. Let $i > 0$, and assume that the statement of the lemma is true for any positive integer smaller than $i$. Suppose that $H_p$ is a chunk of level $i$ that was obtained from a chunk $H$ of level $i-1$ in one recursive step, that is, $\mu(H_p) = i=\mu(H) + 1$.

$H_p$ was obtained by removing a source node $f_s$ from $H$, which is a node of the backbone of $H$. By the inductive hypothesis, all the backbones in $H$ whose level is at most $i-1$ are edge-disjoint and belong to a path $P$ in $H$. Since $f_s \in P$ is a source node in $H$, $f_s$ must be the tail of $P$. If $H_p$ does not contain any node from a backbone whose level is smaller than $\mu(H_p)$, then the entry point $f_{h_p}$ of $H_p$ is defined to be the node in $H_p$ with the minimum index that is adjacent to $f_s$, and in this case $f_{h_p}$ exists. Moreover, there is only one backbone in $H_p$. Therefore, the statement of the lemma is true for $H_p$ in this case. Suppose now that $H_p$ contains at least one node from a backbone whose level is smaller than $\mu(H_p)$. Because the underlying graph of $H_p$ is connected and $P^-=P-f_s$ is a path, it follows that $H_p$ contains $P^-$. Let $b$ be the node adjacent to $f_s$ on $P^-$, \ie, the tail of $P^-$. The path $P^-$ contains all the backbone nodes in $H_p$ of levels smaller than $\mu(H_p)$, and in particular, $P^-$ contains the node $f_{min}$ in $H_p$ of minimum index (the minimum index is with respect to $F$) that belongs to a backbone of level smaller than $\mu(H_p)$. Since $b$ is adjacent to $f_s$, it follows from the preceding that node $f_{h_p}$ exists because $b$ satisfies that it is adjacent to $f_s$ and there is a path from $b$ to $f_{min}$ (which is possibly $b$ itself). Now let $B_{H_p}$ be the backbone of $H_p$. Since $B_{H_p}$ is a path whose head is $f_{h_p}$, and since --- by the choice of $f_{h_p}$ --- there is a path from $f_{h_p}$ to $f_{min}$, the indices of the nodes in $B_{H_p}$ are not larger than the index of $f_{min}$, which is the backbone node on $P^-$ of the minimum index. Therefore, the set of edges in $B_{H_p}$ is disjoint from the set of edges on $P^-$ that belong to backbones of levels smaller than $\mu(H_p)$. Since $B_{H_p}$ is a path in $H_p$, and since there is a path from $f_{h_p}$ to $f_{min}$ in $H_p$, all the backbones in $H_p$ of levels at most $i=\mu(H_p)$ belong to a path in which all the (remaining) nodes of each backbone appear consecutively. This completes the inductive proof. \qed
\end{proof}

\begin{corollary}\label{cor:properties}
All the backbones defined in the recursive definition of $J_C$ are edge-disjoint.
\end{corollary}

\begin{proof}
Let $B$ be a backbone of level $i > 0$ of a chunk $H$, defined during the recursive definition of $J_C$. It suffices to show that the edges of $B$ are different from the edges of all backbones of levels at most $i$. Clearly, the edges of $B$ are different from those of backbones of levels at most $i$ in chunks other than $H$, and from the edges of backbones of levels smaller than $i$ that have been removed prior to the definition of $B$, during the recursive definition of $J_C$. By Lemma~\ref{lem:properties}, the edges of $B$ are different from those of backbones other than $B$ of levels at most $i$ that (may) exist in $H$. The statement follows. \qed
\end{proof}

\begin{lemma}\label{lem:spanning}
Let $C$ be a component of ${\cal D_F}$. The subgraph $J_C$ formed by applying Definition~\ref{def:tree} to $C$ is a spanning subgraph of the underlying graph of $C$.
\end{lemma}

\begin{proof}
The statement follows from the connectedness of $C$ and Definition~\ref{def:tree} by a simple inductive argument: $J_C$ contains a source node $f_s$ of $C$ and an edge from $f_s$ to each chunk in $C - f_s$. Assuming inductively that applying (the recursive definition) Definition~\ref{def:tree} to a chunk of $C-f_s$ results in a spanning subgraph of that chunk, the subgraph obtained by applying Definition~\ref{def:tree} to $C$ results in the spanning subgraph $J_C$ of $C$. \qed
\end{proof}

\begin{lemma} \label{lem:sharetriangle}
Suppose that $H \neq C$ is a chunk obtained from $C$ by removing a sequence of nodes $F^-$ during the construction of $J_C$ in Definition~\ref{def:tree}. Let ${\cal T}$ be the triangulation resulting after performing the sequence of flips $F^-$. Let $f_i$, $f_h$ be the connecting and entry points of $H$, respectively. Then either $\phi(f_i) = \epsilon(f_h)$, or $\phi(f_i)$ and $\epsilon(f_h)$ share a triangle in ${\cal T}$.
\end{lemma}

\begin{proof}
 Suppose that $\phi(f_i) \neq \epsilon(f_h)$ and we show that $\phi(f_i)$ and $\epsilon(f_h)$ share a triangle in ${\cal T}$, which proves the lemma.

    Since $H$ is a chunk of $C$ obtained by removing the sequence of nodes $F^-$, the sequence $F^-$ is a prefix of a topological sorting of ${\cal D}_F$, and hence of $C$. Let $\pi(F)$ be any topological sorting of ${\cal D}_F$ of which $F^-$ is a prefix. By Lemma~\ref{lem:ts}, the DAGs ${\cal D}_F$ and ${\cal D}_{\pi(F)}$ are the same. Since there is an arc from $f_i$ to $f_h$ in $C$ (by the way $f_h$ was chosen), there is an arc from $f_i$ to $f_h$ in ${\cal D}_{\pi(F)}$, and $\phi(f_i)$ and $\epsilon(f_h)$ share a triangle $T$ in the triangulation preceding flip $f_h$, when the flips are performed (starting from ${\cal T}_{initial}$) with respect to the sequence $\pi(F)$. We show next that $T$ is in ${\cal T}$.

  Suppose that this is not the case. Since $f_i$ is the entry point of $H$, $H$ was formed (according to Definition~\ref{def:tree}) upon the removal of node $f_i$, and hence, $f_i$ is a flip in $F^-$. Since $T$ is not in  ${\cal T}$, the aforementioned statement implies that there is a flip $f_p$, between $f_i$ and $f_h$ in $\pi(F)$, that created $T$ for the last time. Hence, both $\phi(f_i)$ and $\phi(f_p)$ belong to $T$ in the triangulation ${\cal T'}$ after flip $f_p$ is performed, following the sequence $\pi(F)$. See Fig.~\ref{fig:adjacent} for illustration. Observe that $\phi(f_i) \neq \phi(f_p)$ because $f_i$ has an arc to $f_h$ (otherwise, $f_i$ will not have an arc to $f_h$, but $f_p$ will have instead). Since $\phi(f_i)$ and $\phi(f_p)$ share triangle $T$ in ${\cal T'}$, $\phi(f_i)$ and $\epsilon(f_p)$ also share a triangle in the triangulation before flip $f_p$ is performed, following the sequence $\pi(F)$. This implies that $f_i$ has an arc to $f_p$ in ${\cal D}_{\pi(F)}$, and since ${\cal D}_{\pi(F)}={\cal D}_F$, there is an arc from $f_i$ to $f_p$ in $C$.
 Since both $\phi(f_p)$ and $\epsilon(f_h)$ (possibly identical) belong to triangle $T$ in ${\cal T'}$, and by the choice of $f_p$, there is an arc from $f_p$ to $f_h$ in the component $C$ of ${\cal D}_{\pi(F)}={\cal D}_F$. Since $H$ is a chunk, and hence is connected, and since there is an arc from $f_i$ to $f_p$ and an arc from $f_p$ to $f_h$, upon the removal of $f_i$ to form $H$, $f_p$ was a node in $H$. This, however, contradicts the choice of the entry point $f_h$ of $H$ (in either of the two cases used to define the entry point) because $f_i$ is adjacent to $f_p$ and there is an arc from $f_p$ to $f_h$ (which implies that $f_p$ has a smaller index than $f_h$ in $F$). \qed

\end{proof}

We describe next the sequence of actions, ${\cal S}$, claimed in Theorem~\ref{thm:traverse}, that \algo~performs starting at a changed edge corresponding to a node in $C$. The sequence ${\cal S}$, which we will recursively describe, corresponds to a recursive traversal of $J_C$.  Let $f_i$, $f_h$ be the connecting and the entry points to a chunk $H \neq C$, respectively. At the top level of the recursive definition, where $H=C$ is a component of ${\cal D}_F$, define $f_h$ to be a flip in $C$ whose underlying edge $\epsilon(f_h)$ is a changed edge ($f_i$ need not be defined). We initialize ${\cal S}$ to be empty. The recursive construction of ${\cal S}$, denoted \rec, is described schematically in {\bf Algorithm 2} below. \\

\setcounter{figure}{1}
\renewcommand{\figurename}{Algorithm}
 \begin{algorithmnew}{The sequence of actions for a chunk $H$.}{alg:nondeterministic}
\algtitle{\rec}

Let $H$ be a chunk with entry point $f_h$. The sequence of actions on $H$ is defined as follows.
\begin{codeblock}
\step Append to ${\cal S}$ a sequence of actions of type (i) that performs the walk $W$ from $\epsilon(f_h)$ to $\epsilon(f_s)$ (in the current triangulation ${\cal T}$) defined in Lemma~\ref{lem:walk}, associated with the backbone $B=\langle f_s=b_1, \ldots, b_{\ell}=f_h\rangle$ of $H$.

\step Let $H^-=H-f_s$. If $H^- = \emptyset$ then append to ${\cal S}$ a type-(iv) action that flips $\epsilon(f_s)$ and moves to the edge on the top of the stack $\Lambda$. Else do the following:

\begin{itemize}
\step[2.1.] Let $H_1, \ldots, H_x$ be the components of $H^-$.

\step[2.2.] If $x > 1$ then append to ${\cal S}$ a type-(iii) action that flips $\epsilon(f_s)$, pushes $\phi(f_s)$ into the stack $\Lambda$, and moves to the underlying edge of the entry point of $H_1$; else ($x=1$) append to ${\cal S}$ a type-(ii) action that flips $\epsilon(f_s)$ and moves to the underlying edge of the entry point of $H_1$.


\step[2.3.] For $p=1, \ldots, x$:

\begin{itemize}
\step[2.3.1.] Append to ${\cal S}$ the sequence of actions constructed by recursing on $H_p$. For the last action of $H_p$, which must be a flip-action (flipping the underlying edge of the last flip traversed in $H_p$), do the following:

\begin{itemize}
\step[2.3.1.1.] if $x > 1$ and $p < x$ (there is more than one component in $H^-$ and $H_p$ is not the last component of $H^-$ considered in the for loop in step 2.3), then the last action of $H_p$ appended to ${\cal S}$ is a type-(iv) action that moves to the edge on the top of the stack $\Lambda$ after performing the last flip in $H_p$;
\step[2.3.1.2.] else if $x > 1$ and $p=x$ (there is more than one component in $H^-$ and $H_p$ is the last component  considered in the for loop in step 2.3) then the last action of $H_p$ appended to ${\cal S}$ is a type-(v) action that moves to the edge on the top of the stack $\Lambda$ and pops $\Lambda$ after performing the last flip in $H_p$.
\end{itemize}
\end{itemize}
\end{itemize}
\end{codeblock}
\end{algorithmnew}

\renewcommand{\figurename}{Fig.}
\setcounter{figure}{3}
\begin{figure}
\begin{center}
 \resizebox{1\textwidth}{!}{
\begin{tikzpicture}
[vertex/.style={fill,circle,inner sep = 2pt}]

\node (a) at (0, 0) [vertex]{};
\node (b) at (6,6) [vertex] {};
\node (c) at (12,0) [vertex]{};
\node (d) at (8, -3.4) [vertex] {};
\node (e) at (-4, 3) [] {};

\draw[line width=2] (a) -- (b) -- (c) -- (a);
\draw[line width=0.3, dashed] (b) -- (d);

\node at (2, 3) {\Large $\phi(f_i)$};
\node at (4, -0.5) {\Large $\phi(f_p)$};

\filldraw[fill=lightgray, draw=lightgray] (3.5, 1.6) circle[radius= 0.5cm] node {\Large $f_i$};
\filldraw[fill=lightgray, draw=lightgray] (9, 3) circle[radius= 0.5cm] node {\Large $f_h$};
\filldraw[fill=lightgray, draw=lightgray] (7.7, -1.8) circle[radius= 0.5cm] node {\large $f_p$};

\draw[->,lightgray, line width=2, dashed] (3.8, 1.2) -- (7.3, -1.3);
\draw[->,lightgray, line width=2, dashed] (7.7, -1.3)  -- (8.8, 2.4);

\end{tikzpicture}
}
\end{center}
\caption{Illustration for the proof of Lemma~\ref{lem:sharetriangle}. The black vertices and black solid edges are vertices and edges of ${\cal T}$, respectively.  The gray circles labelled by flips represent the nodes/flips in $C$, and the dashed (thick) gray edges are arcs in $C$. The dashed black edge, through the node labeled $f_p$, is the edge (in a previous triangulation) $\epsilon(f_p)$ that was flipped into $\phi(f_p)$; the solid black edge through $f_h$ is $\epsilon(f_h)$.}
\label{fig:adjacent}
\end{figure}
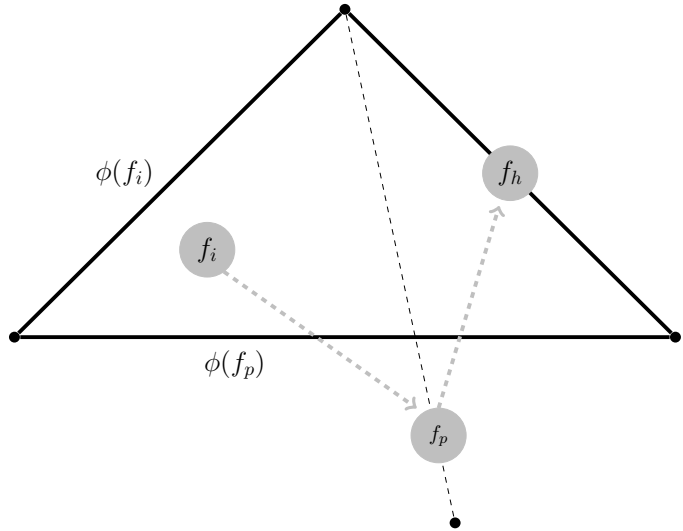

We are now ready to prove Theorem~\ref{thm:traverse}.

\begin{proof} ({\bf Theorem~\ref{thm:traverse}})
Let ${\cal S}$ be the sequence of actions constructed according to \rec~when invoked on component $C$, which corresponds to a traversal of $J_C$. It is clear that the order of the flips in ${\cal S}$ corresponds to a topological sorting of $C$ because every flip corresponds to the removal of a source node from a chunk resulting from $C$ in the recursive definition of $J_C$, and because $J_C$ is a spanning subgraph of $C$ by Lemma~\ref{lem:spanning}. Therefore, it suffices to show the validity of the construction of ${\cal S}$ (\ie, that the actions taken in the construction are valid actions), and that the length of ${\cal S}$ is at most $11|V(C)|$. To do the latter, we charge the actions of ${\cal S}$ to the nodes and edges of $J_C$.

\rec~is invoked on a chunk $H$ with an entry node $f_h$ of $H$ (initially $H=C$ and $f_h$ is a node in $C$ whose underlying edge is a changed edge in the current triangulation ${\cal T}$). In Lemma~\ref{lem:walk} we showed that there is a path $B=\langle f_s=b_1, \ldots, b_{\ell}=f_h\rangle$ from a source node $f_s$ in $H$ to $f_h$ that corresponds to a walk $W$ from edge $\epsilon(f_h)$ to $\epsilon(f_s)$ in the current triangulation ${\cal T}$; moreover, the length of this walk is at most the length of $B$. \rec~appends to ${\cal S}$ a sequence of actions for performing $W$ of type (i) (as defined in Subsection~\ref{subsec:actions}), and the number of such actions is at most the length of $B$. Then \rec~appends to ${\cal S}$ an action that flips $\epsilon(f_s)$, which is one action either of type (ii), (iii) or (iv). Next, \rec~recurses on each chunk $H_p$ of $H-f_s$, starting at the entry point $f_{h_p}$ of $H_p$. In Lemma~\ref{lem:sharetriangle}, we showed that the edges $\phi(f_s)$ and $\epsilon(f_{h_p})$ are either identical, or they share a triangle in the current triangulation when \rec~is recursively called on $H_p$. Hence, in at most one action of type-(i), we can move from $\phi(f_s)$ to $\epsilon(f_{h_p})$. If there is more than one chunk in $H-f_s$, \rec~appends an action of type-(iii) to ${\cal S}$ that pushes $\phi(f_s)$ into the stack after flipping $f_s$; moreover, the last action that it appends to ${\cal S}$ after recursing on the components of $H^-$ is an action of type-(v) that pops the stack and moves to the top of it.

To prove that the length of ${\cal S}$ is at most $11|V(C)|$, we charge the actions in ${\cal S}$ to the nodes and edges of $J_C$. The actions in ${\cal S}$ can be classified into two categories: flip-actions and move-actions. The number of flip-actions is at most the number of nodes in $J_C$, which is $|V(C)|$. Note that actions that involve moving to the top of the stack, or popping the stack, or both, are combined with flips, and hence have been accounted for. The move-actions are all of type (i), and can be further divided into two groups: (I) those corresponding to a walk $W$ associated with a backbone $B$ of a chunk $H$, and (II) those corresponding to moves from an edge $\phi(f_s)$ (on the top of the stack), corresponding to a source $f_s$ in a chunk $H$, to an edge whose corresponding node is an entry point of a chunk in $H-f_s$. The number of actions in group (I) is at most $|E(C)|$; this is because, by Corollary~\ref{cor:properties}, the edges of different backbones are distinct, and hence the total number of such edges (and hence actions in group (I)) is at most $|E(J_C)| \leq |E(C)|$. To bound the number of actions in group (II), observe that each such action corresponds to an edge in $C$ from $f_s$ to the entry point of a chunk resulting from removing $f_s$ from $H$. Since $f_s$ is removed from $J_C$ upon making the recursive calls to the resulting chunks, we can charge each such action in a one-to-one fashion to edges of $E(J_C)$. Therefore, the number of actions in group (II) is at most  $|E(J_C)| \leq |E(C)|$, and the total number of actions of type (i) is at most $2|E(J_C)| \leq 2|E(C)|$. It follows that the length of ${\cal S}$ is at most $|V(C)| + 2|E(J_C)| \leq 11 |V(C)|$ (by Lemma~\ref{lem:dagsize}). \qed
\end{proof}


We conclude with the following main theorem:

\begin{theorem}\label{thm:final}
Let $({\cal T}_{initial}, {\cal T}_{final}, k)$ be an instance of {\sc Flip Distance}. The algorithm \algo~decides $({\cal T}_{initial}, {\cal T}_{final}, k)$ correctly in time $\Oh(n+k\cdot c^{k})$, where $n$ is the number of points in the triangulations.
\end{theorem}

\begin{proof}
It is easy to see that the correctness of \algo \\ follows from: (1) there is an enumeration in step 2 of the algorithm of the correct number of components $t$, and of $(k_1, \ldots, k_t)$ such that $k_i$ is the exact number of flips in $C_i$, $i=1, \ldots, t$; and (2) the existence of a sequence ${\cal S}$ (by Theorem~\ref{thm:traverse}) of actions of length at most $11k_i$ that, starting from any changed edge in $C_i$, performs the $k_i$ flips in $C_i$ in a topologically-sorted order. We only need to analyze the running time of the algorithm.

The initial processing of the triangulations to find the changed edges takes $\Oh(n)$ time because the number of edges in any triangulation of ${\cal P}$ is $\Oh(n)$. The total number of sequences $(k_1,\ldots, k_{t})$, for $t=1, \ldots, k$, satisfying $k_1+\cdots +k_t \leq k$ and $k_1, \ldots, k_t \geq 1$, is at most $2^k$. This is because the total number of sequences $(k_1,\ldots, k_{t})$, for $t=1, \ldots, k$, satisfying $k_1+\cdots +k_t = k$ is the {\em composition number} of (integer) $k$, which is equal to $2^{k-1}$; the above upper bound of $2^k$ now follows by summing over all positive integers less $\leq k$. For each such sequence $(k_1, \ldots, k_t)$, the algorithm iterates through the numbers $k_1, \ldots, k_t$ in the sequence. For a number $k_i$, $1 \leq i \leq t$, \algo~branches into all possible sequences ${\cal S}$ of length at most $11k_i$, and in which each action is one of 14 choices (by Proposition~\ref{prop:choices}). Therefore, the number of such sequences is at most $14^{11k_i}$. It follows that the total size of the search trees depicted by the algorithm to find a solution (if it exists) is at most:
$
\sum_{t=1}^k\sum_{\begin{subarray}{c}\text{$(k_1, \ldots, k_t)$}
     \end{subarray}} (14^{11k_1} \times \cdots \times 14^{11k_t}) = \Oh(2^{k} 14^{11k})=\Oh(c^k),\nonumber\label{ineq:bound}
$ where $c \leq 2 \cdot 14^{11}$. Since the actions along any root-leaf path in the search tree can be carried out in time $\Oh(k)$, and the resulting triangulation at the end of the path can be compared to ${\cal T}_{final}$ in $\Oh(k)$ time as well (because the number of flips performed is at most $k$), the running time along each root-leaf path in the tree is $\Oh(k)$. (Note that, at a given edge in the triangulation, we can test whether flipping the edge is admissible in $\Oh(1)$ time by testing whether or not the quadrilateral determined by the two triangles sharing the edge is convex.)  It follows from the above that the running of \algo~is $\Oh(n+k\cdot c^k)$.  \qed
\end{proof}

\begin{remark}\label{rem:equalflipdistance}
We note that \algo, within the same time upper bound, can be used to decide whether $k$ is the flip distance between ${\cal T}_{initial}$ and ${\cal T}_{final}$, which means that no sequence of flips of length smaller than $k$ exists that transforms ${\cal T}_{initial}$ to ${\cal T}_{final}$. This is because the algorithm is trying all sequences $(k_1,\ldots, k_{t})$, for $t=1, \ldots, k$, satisfying $k_1+\cdots +k_t \leq k$ and $k_1, \ldots, k_t \geq 1$. If $({\cal T}_{initial}, {\cal T}_{final}, k')$, where $k'=k_1+\cdots +k_t$,  is a yes-instance only for $k'=k$, then the flip distance between ${\cal T}_{initial}$ and ${\cal T}_{final}$ is $k$.
\end{remark}

\section{Extensions to other triangulations}
\label{sec:extensions}

\subsection{Triangulated polygonal regions with holes}
\label{subsec:polygonalregions}
The triangulations considered in this paper are triangulations of a point set. It is not difficult, however, to see that the presented algorithm works as well for triangulations of any polygonal region, even with
holes in its interior (note that a single point is considered as a degenerate hole). This can be easily seen in one of two ways. The first is to observe that the presented algorithm flips an edge only if the flip is admissible, where flipping an edge is admissible if the edge is a diagonal in a convex quadrilateral formed by two adjacent triangles in the triangulation. By modifying the definition of the admissibility of a flip to exclude the fixed edges that form the boundary of the polygonal region (\ie, flipping these edges is inadmissible), which subsume the boundary edges of holes, the presented algorithm will seamlessly work within the same time upper bound on triangulations of a polygonal region with holes. The second way the presented algorithm can be used to solve the {\sc Flip Distance} problem on triangulations of a polygonal region with holes, is to use the reduction in~\cite{lubiw} from the {\sc Flip Distance} of triangulations of a polygonal region with holes to the {\sc Flip Distance} of triangulations of a point set; this reduction is a polynomial-time reduction that preserves the flip distance value. Using this reduction, however, will result in a polynomial-time increase in the overall running time. The idea behind the reduction in~\cite{lubiw} is to take two triangulations ${\cal T}_{initial}$ and ${\cal T}_{final}$ of a polygonal region (with holes) and reduce them in polynomial time to two triangulations ${\cal T'}_{initial}$ and ${\cal T'}_{final}$ of a point set, such that the flip distance between ${\cal T}_{initial}$ and ${\cal T}_{final}$ is the same as that between ${\cal T'}_{initial}$ and ${\cal T'}_{final}$. To do so, each boundary edge of a hole is duplicated $n^2$ times, where $n$ is the number of vertices of the polygonal region, using a gadget. This duplication is done carefully so that ``dismantling'' a gadget does not offer any savings in the flip distance, and hence, the flip distance is preserved. This implies that the $\Ostar(c^{k})$-time algorithm presented in this paper can be composed with the reduction in~\cite{lubiw} to give an $\Ostar(c^{k})$-time algorithm for the {\sc Flip Distance} problem on triangulations of polygonal regions with holes, albeit with an additive quadratic (in $n$) increase in the overall running time of the algorithm in this paper, if this second approach is adopted.

\subsection{Triangulated graphs}
\label{subsec:planar}
In this section we show how to extend our results to labeled triangulated graphs. Without loss of generality, we assume throughout this section that the number of vertices in the graphs under consideration is at least 4.
We consider the problem of computing the flip distance between two labeled triangulated graphs, that is, two maximal plane triangulations defined on the same vertex set, or via a bijection between their vertex sets.
In this setting, referred to as the {\em combinatorial setting} as opposed to the {\em geometric setting} treated in the previous section, an edge in an embedding of a triangulated graph $G$ no longer needs to be drawn as a line segment, and can be drawn as an arc (curve) joining the endpoints of the edge. A flip of an edge in this setting is the operation of replacing the edge with another edge to obtain another triangulated graph, provided that the obtained graph remains simple (no multiple edges are created), otherwise, flipping the edge is not admissible; such a flipping operation is called a {\em combinatorial flip}, as opposed to the {\em geometric flip} considered in the previous section. The difference between the geometric and combinatorial settings is that there is no convexity requirement, and that edges can be drawn as arcs as opposed to line segments. In particular, while every admissible geometric flip is also an admissible combinatorial flip, the converse may not be true (see Fig.~\ref{fig:combvsgeo1} and Fig.~\ref{fig:combvsgeo2}). The flip distance problem in the combinatorial setting is defined as follows: \\

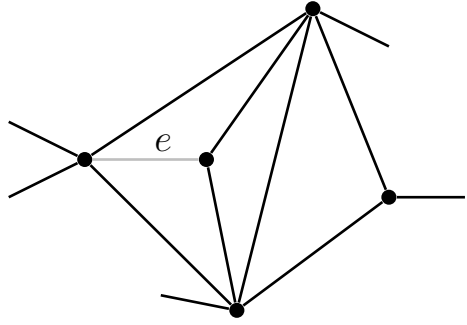
\begin{figure}
\centering
\begin{tikzpicture}
[vertex/.style={fill,circle,inner sep = 2pt}]

\node (x) at (7, 1) [vertex]{};
\node (a) at (4,-1) [vertex] {};
\node (b) at (5.6,-1) [vertex]{};
\node (c) at (8,-1.5) [vertex] {};
\node (d) at (6,-3) [vertex] {};

\draw[line width=1] (a) -- (x) -- (c) -- (d) -- (a);

\draw[line width=1] (3, -1.5) -- (a) -- (3, -0.5); \draw[line width=1] (c) -- (9, -1.5); \draw[line width=1] (x) -- (8, 0.5); \draw[line width=1] (d) -- (5, -2.8);

\draw[black, line width=1] (x) -- (d);
\draw[black, line width=1] (x) -- (b) ;
\draw[black, line width=1] (b) -- (d) ;
\draw[lightgray, line width=1] (b) -- (a) node[midway, anchor=south west] {\black \Large $e$}; \end{tikzpicture}
\caption{Neither a geometric nor a combinatorial flip to $e$ is admissible.}
\label{fig:combvsgeo1}
\end{figure}

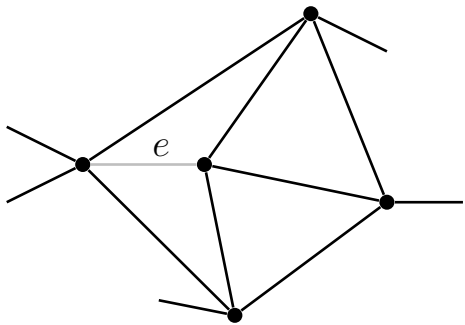
\begin{figure}
\centering

\begin{tikzpicture}
[vertex/.style={fill,circle,inner sep = 2pt}]

\node (x) at (7, 1) [vertex]{};
\node (a) at (4,-1) [vertex] {};
\node (b) at (5.6,-1) [vertex]{};
\node (c) at (8,-1.5) [vertex] {};
\node (d) at (6,-3) [vertex] {};

\draw[line width=1] (a) -- (x) -- (c) -- (d) -- (a);

\draw[line width=1] (3, -1.5) -- (a) -- (3, -0.5); \draw[line width=1] (c) -- (9, -1.5); \draw[line width=1] (x) -- (8, 0.5); \draw[line width=1] (d) -- (5, -2.8);

\draw[black, line width=1] (b) -- (c);
\draw[black, line width=1] (x) -- (b) ;
\draw[black, line width=1] (b) -- (d) ;
\draw[lightgray, line width=1] (b) -- (a) node[midway, anchor=south west] {\black \Large $e$}; \end{tikzpicture}
\caption{A combinatorial flip to $e$ is admissible, but a geometric flip is not.}
\label{fig:combvsgeo2}
\end{figure}

\paramproblem{{\sc Combinatorial Flip Distance}}{Two labeled triangulated graphs $G$ and $H$.}{$k$.}{Is the minimum number of flips needed to transform $G$ into a graph isomorphic to $H$ at most $k$?}

It is known that the flip distance between any two labeled triangulated graphs is $\Oh(n\lg{n})$~\cite{sleator}.

To extend the algorithm for \FD, presented in the previous section, to \CFD, we need to show that the results that rely on geometric arguments, which the algorithm in turn relies on, hold in the combinatorial setting as well. While the general skeleton of the algorithm remains the same, the underlying results showing the correctness of the algorithm need to hold true in the combinatorial setting. This is not obvious because in the geometric setting there is a unique realization (\ie, drawing) of an edge between two points $x$ and $y$, namely the line segment joining them, whereas in the combinatorial setting, there are infinitely many ways of realizing an edge between $x$ and $y$.  Even a local deformation of a plane triangulation changes the embedding, but not the graph or the rotation system. The structural results essential for proving the correctness of the algorithm (the results in Section~\ref{sec:structural}  and Lemmas~\ref{lem:walk} and~\ref{lem:sharetriangle}) rely on arguments that assume that an edge is realized as a line segment. Therefore, we need to modify these results so that they work in the combinatorial setting, in which an edge may not be realizable as a line segment. For instance, Lemma~\ref{lem:crossing} relies on edge-crossing arguments that hold in the geometric setting but not necessarily in the combinatorial one. To circumvent these issues, we use the notions of isotopy to define equivalence between different realizations of an edge, and the notion of crossing number to define when two edges cross. We show how using these notions the structural results needed to prove the correctness of the algorithm in the geometric setting carry over to the combinatorial setting.

We start by giving the terminologies and definitions needed to extend the results to the combinatorial setting. Let $\VVV$ be a set of points in the plane.

\begin{definition} \rm
An {\em edge} between two points $a, b \in \VVV$ is an arc embedded in the plane that meets the vertex set $\VVV$ precisely in its endpoints $a$ and $b$.
\end{definition}

\begin{definition}\rm
An {\em edge isotopy} between two edges $e$ and $e'$ is a continuous deformation through edges from $e$ to $e'$ (\ie, no intermediate edge in this continuous deformation meets a vertex in its interior).  It follows from continuity that $e, e'$, and every intermediate edge, meet $\VVV$ precisely in the same pair of endpoints as $e, e'$. Two edges $e, e'$ are {\em isotopic}, denoted $e \sim e'$,  if there is an edge isotopy between them.
\end{definition}

Let $G$ be a labeled triangulated graph. Since $G$ is 3-connected, it has a unique plane embedding up to homeomorphism~\cite{diestel}. Recall that each face in the embedding of $G$, including the outer face, is a triangle.  Consider any embedding $\eta(G)$ of $G$, and let $\VVV$ be the set of (labeled) points/vertices of $\eta(G)$. Let $e=xy$ be an edge of $\eta(G)$. Similarly to the geometric setting, we define the {\em quadrilateral associated with} $e$, $Q_e$, to be the quadrilateral formed by the two faces of $\eta(G)$ that share $e$.
Removing $e$ from $Q_e$ creates a 4-face $\gamma=(x, a, y, b)$.   A {\em flip} $f$ with underlying edge $e$ is the operation of removing $e$ and replacing it with an edge between $a$ and $b$ to obtain another plane triangulation, provided that there is no edge between $a$ and $b$ in $G$; otherwise, flipping $e$ is not admissible. Observe that if $f$ is admissible then the edge between $a$ and $b$ replacing $e$ must be interior to $\gamma$, and that if $f$ is not admissible and an edge $ab$ already exists then this edge must lie outside $\gamma$. Observe also that any two edges between $a$ and $b$ that are interior to $\gamma$ are isotopic. Therefore, any two embeddings that can be obtained from $f$ by two realizations of an edge $ab$ are homeomorphic. As before, $\epsilon(f)$ denotes the underlying edge $e$ of flip $f$ and $\phi(f)$ denotes the new edge resulting from performing $f$.

Suppose that we flip an edge $e$ in $G$ to obtain a triangulated graph $G'$, and let $\eta(G)$ and $\eta(G')$ be any plane embeddings of $G$ and $G'$, respectively. Since both $G$ and $G'$ are 3-connected, and hence, each has a unique plane embedding, it follows that if we perform a flip $f$ --- as described above --- to the realization of $e$ in $\eta(G)$ we obtain an embedding that is homeomorphic to $\eta(G')$. On the other hand, because $G'$ has a unique embedding, a plane embedding of a graph is homeomorphic to $\eta(G')$ if and only if this embedding is isomorphic to $G'$. Therefore, to solve \CFD, we can equivalently compute a plane embedding $\eta(G)$ of $G$ and ask whether there exists a sequence of $k$ flips that transforms $\eta(G)$ to
an embedding that is isomorphic to $H$. We redefine \CFD~so it becomes amenable to our approach: \\

 \paramproblem{{\sc Combinatorial Flip Distance}}{A labeled plane triangulation $G$ and a labeled triangulated graph $H$.}{$k$.}{Is the minimum number of flips needed to transform $G$ to a graph isomorphic to $H$ at most $k$?}

For ease of the notation, and since in the remainder of this section we will be mainly dealing with an embedding of $G$, we abused the notation in the above problem definition and used $G$ to refer to an embedding of the triangulated graph $G$. (Again, because $G$ is 3-connected, there is a correspondence between the graph $G$ and its unique embedding or plane triangulation.)  We will assume that the points of $G$ are labeled with the same labels as the vertices of $H$. The locations of these points will be unchanged by any sequence of flips.

\begin{definition}\rm
Let $e, e'$ be two edges that may or may not be disjoint in their interiors.  The {\em intersection number} between $e$ and $e'$, $i(e,e')$, is the minimal number of intersections taken over pairs of isotopic edges, \ie,  $i(e,e') = \min \{ | e_0\cap e_0'| : e_0 \sim  e, e_0' \sim e'\}$, where $| e_0\cap e_0'|$ is the number of intersection points between the interior of $e_0$ and the interior of $e_0'$.
\end{definition}

If $i(e,e')>0$, then the edges $e$ and $e'$ cross each other inextricably, that is, no edge isotopy can separate them.

 Let $G$ be a plane triangulation, and let $F=\langle f_1,\ldots, f_r \rangle$ be a realization of a valid sequence of flips (\ie, corresponding to a realization of the edges resulting from the flips) with respect to $G$. Observe that any two realizations of the same sequence of flips yield two homeomorphic embeddings. We denote by $G_i$, $i=1, \ldots, r$, the plane triangulation resulting from applying the sequence
 $\langle f_1,\ldots, f_i \rangle$ of flips to $G$; we define $G_0=G$. The notion of flip adjacency is defined in exactly the same way as in Definition~\ref{def:adjacent}. We note that replacing equality of edges with isotopy in Definition~\ref{def:adjacent} yields the same outcome because $f_i \rightarrow f_j$ implies that $f_i$ is the last flip that creates $\epsilon(f_j)$ or an edge of $Q_{\epsilon(f_j)}$.

 Let $(G, H, k)$ be an instance of \CFD. By a solution to the instance $(G, H, k)$ we mean a realization of a sequence of at most $k$ flips that transforms $G$ into an embedding isomorphic to $H$. Assume that $F=\langle f_1,\ldots, f_r \rangle$, where $r \leq k$, is a solution to $(G, H, k)$. The definition of the DAG ${\cal D}_F$ is the same as that in Section~\ref{sec:structural},  and
 Lemma~\ref{lem:dagsize} follows using the same proof.  Lemma~\ref{lem:ts} and Corollary~\ref{cor:ts} also follow using the same arguments.

Since the goal is to transform $G$ into an embedding that is {\em isomorphic} to $H$, we can define changed edges the same way we did in the geometric setting: An edge $e$ is a \emph{changed edge} if $e \in G$ but there is no edge between the endpoints of $e$ in $H$. The definition of essential and nonessential components is the same as well, and Lemma~\ref{lem:essential} carries over in the same way.

 The statements of Lemmas~\ref{lem:crossing}~--~\ref{lem:path} need to be changed to use the terminology of isotopy and intersection number, and the proofs of some of them need to be modified. We explain below how they need to be changed.

Both the statement and the proof of Lemma~\ref{lem:crossing} need to be modified. Since this modification is nontrivial, we present the proof for the sake of completeness.
Let $(G, H, k)$ be an instance of \CFD. Assume that $F=\langle f_1,\ldots, f_r \rangle$, where $r \leq k$, is a normalized solution to the instance, and let $C$ be a component of ${\cal D}_F$.

 \begin{lemma} \label{lem:crossingcombo}
Let $f_i$ and $f_h$, where $i < h$ and $\epsilon(f_i) \not \sim \epsilon(f_h)$, be two flips in $C$ such that $i(\phi(f_h), \epsilon(f_i)) > 0$, and $\epsilon(f_i)$ is not flipped between $f_i$ and $f_h$. There is a directed path from $f_i$ to $f_h$ in $C$.
\end{lemma}

\begin{proof}
Assume that the statement is not true, and let $f_i$ and $f_h$ be the closest pair of flips in $C$ (in terms of the number of flips between $f_i$ and $f_h$) satisfying the conditions in the statement of the lemma and such that there is no directed path from $f_i$ to $f_h$ in $C$.

Consider the quadrilateral $Q_{\phi(f_h)}$ associated with $\phi(f_h)$ in ${\cal T}_h$. (See Fig.~\ref{fig:crossingcombo} for illustration.) Since $\epsilon(f_i) \not \sim \epsilon(f_h)$ and $i(\phi(f_h), \epsilon(f_i)) > 0$, $\epsilon(f_i)$ cannot be a diagonal or an edge of $Q_{\phi(f_h)}$: it is not a diagonal because
$\epsilon(f_i) \not \sim \epsilon(f_h)$ and $i(\phi(f_h), \epsilon(f_i)) > 0$, and it is not an edge because $i(\phi(f_h), \epsilon(f_i)) > 0$. Because $i(\phi(f_h), \epsilon(f_i)) > 0$, using the Jordan curve theorem it can be easily seen that at least one edge $e$ of $Q_{\phi(f_h)}$ satisfies $i(e, \epsilon(f_i)) > 0$. Therefore, there must exist a flip $f_p$, where $i \leq p < h$, such that $\phi(f_p)=e$. Since $e$ is an edge of $Q_{\phi(f_h)}$ which contains $\phi(f_h)$ as a diagonal in $G_h$, $e$ is also an edge of $Q_{\phi(f_h)}$ in $G_{h-1}$. This means that $\phi(f_p)$ and $\epsilon(f_h)$ share a face in $G_{h-1}$, and hence there is an arc from $f_p$ to $f_h$ (we can assume that $f_p$ is the last flip before $f_h$ such that $e = \phi(f_p)$). If $i = p$ then there is a path (of length 0) from $f_i$ to $f_p$; on the other hand, if $i < p$, then since $i(\phi(f_p), \epsilon(f_i)) > 0$, $p < h$, and $\epsilon(f_p) \neq \epsilon(f_i)$, by the way $f_h$ is chosen, there is a directed path from $f_i$ to $f_p$ in $C$. Therefore, there is a directed path from $f_i$ to $f_h$ in $C$ --- a contradiction. \qed
\end{proof}

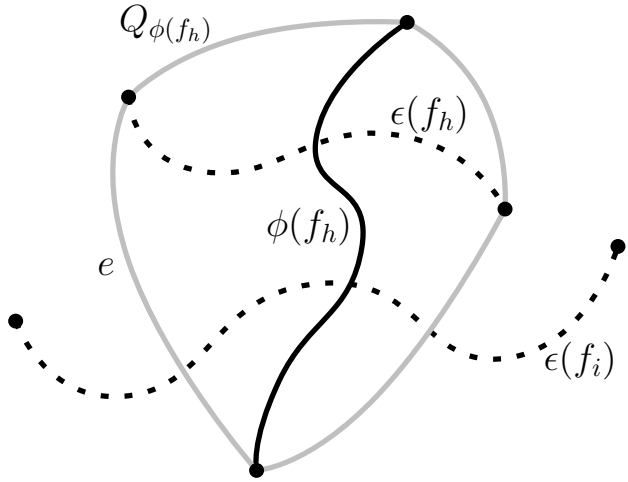
\begin{figure}
\centering
 \resizebox{0.7\textwidth}{!}{
\begin{tikzpicture}
[vertex/.style={fill,circle,inner sep = 2pt}]

\draw[line width=2,loosely dashed] plot [smooth, tension=1] coordinates {(-4, 0)  (-2.5, -1)  (0,0.5)   (2.5, -0.5) (4, 1)};

\draw[line width=2,loosely dashed] plot [smooth, tension=1] coordinates {(-2.5, 3)  (-1.5, 2)    (1, 2.5) (2.5, 1.5)};

\draw[line width=2] plot [smooth, tension=1] coordinates {  (1.2, 4)    (0, 2.5) (0.6, 1) (-0.5, -0.8) (-0.8, -2)};

\draw[lightgray, line width=2] plot [smooth, tension=1] coordinates {(1.2, 4)    (-1, 3.8) (-2.5, 3)};

\draw[lightgray, line width=2] plot [smooth, tension=1] coordinates {(1.2, 4)    (2.2, 3) (2.5, 1.5)};

\draw[lightgray, line width=2] plot [smooth, tension=1] coordinates {(-0.8, -2)    (0.8, -1)  (2.5, 1.5)};

\draw[lightgray, line width=2] plot [smooth, tension=1] coordinates {(-0.8, -2)     (-2.5, 0.8)   (-2.5, 3)};

\node (a) at (-4, 0) [vertex]{};
\node (b) at (4, 1) [vertex]{};
\node (c) at (-2.5, 3)  [vertex]{};
\node (d) at (2.5, 1.5) [vertex]{};
\node (e) at (1.2, 4)  [vertex]{};
\node (f) at (-0.8, -2) [vertex]{};

\node[label=above:{\Large $\epsilon(f_i)$}] at (3.5, -1) {};
\node[label=above:{\Large $\epsilon(f_h)$}] at (1.5, 2.3) {};
\node[label=above:{\Large $\phi(f_h)$}] at (-0.1, 0.8) {};
\node[label=above:{\Large $Q_{\phi(f_h)}$}] at (-2, 3.5) {};
\node[label=above:{\Large $e$}] at (-2.8, 0.4) {};

\end{tikzpicture}
}

\caption{Illustration of the proof of Lemma~\ref{lem:crossingcombo}.}
\label{fig:crossingcombo}
\end{figure}

The remaining lemmas, Lemma~\ref{lem:flip-restore}, Lemma~\ref{lem:flip-flip}, and Lemma~\ref{lem:path} hold true in the combinatorial setting after changing the edge equality in their statements to edge isotopy (see Lemma~\ref{lem:allpathcombinatorial} below), changing edge-crossing to the edge intersection number being greater than 0, and modifying their proofs accordingly. We skip the proofs to avoid repetition.

\begin{lemma}\label{lem:allpathcombinatorial}
Let $f_i$ and $f_h$, where $i < h$, be two flips in $C$. If one of the following conditions is true, then there is a directed path from $f_i$ to $f_h$ in $C$:
\begin{itemize}
\item[(1)] $i(\phi(f_h), \epsilon(f_i)) > 0$.
\item[(2)] $\phi(f_h) \sim \epsilon(f_i)$.
\item[(3)] $\epsilon(f_i) \sim \epsilon(f_h)$.
\item[(4)] $\phi(f_i) \sim \epsilon(f_h)$, or $\phi(f_i)$ (or an edge $e \sim \phi(f_i)$) and $\epsilon(f_h)$ share a triangle in $G_j$, for some $j$ satisfying $i \leq j < h$.
\end{itemize}
\end{lemma}

After the modifications described above, \algo \\ for {\sc Flip Distance} that performs a topological sorting of ${\cal D}_F$ can be adopted for {\sc Combinatorial Flip Distance}, albeit with some minor changes in the notation. We present the algorithm for {\sc Combinatorial Flip Distance}, denoted \algocomb, schematically in {\bf Algorithm 3}, for the sake of completeness. The algorithm takes an instance $(G, H, k)$ of {\sc Combinatorial Flip Distance}.
\renewcommand{\figurename}{Algorithm}
\setcounter{figure}{2}

\begin{algorithmnew}{The algorithm for {\sc Combinatorial Flip Distance}.}{alg:wholealgo}
\algtitle{\algocomb}
\begin{codeblock}
\step Order the changed edges in $G$ arbitrarily and denote this ordering by $\Gamma$.

\step For each $t=1, \ldots, k$, and for each tuple $(k_1, \ldots, k_t)$ satisfying $k_1+k_2+\cdots +k_t \leq k$ do:

\begin{itemize}
\step[2.1.]  Start a new search tree $\Upsilon$, where $G$ is the triangulated graph at the root.

\step[2.2.]  Perform $t$ branching stages as follows, where the branching in stage $\ell$, for $\ell =1$ to $t$, is performed starting from {\em each} leaf-node $\alpha$ of the search tree resulting at the end of stage $\ell-1$ (if $\ell = 1$ then there is only one leaf-node $\alpha$, which is the root of the search tree):
\begin{itemize}
\item[2.2.1.] Let $G_{\alpha}$ be the triangulated graph at node $\alpha$ in $\Upsilon$.
\item[2.2.2.]  Let $e$ be the first edge in $\Gamma$ with respect to the ordering such that $e$ is in $G_{\alpha}$ but not in $H$. Do the following:

 \begin{itemize}
 \item[2.2.2.1.] Associate the triplet $(G_{\alpha}, e, \Lambda=\emptyset)$ with the current node $\alpha$ of the search tree.
 \item[2.2.2.2.] Branch into all possible sequences ${\cal S}$ of actions of length at most $11k_{\ell}$, in which the number of flips is at most $k_{\ell}$. For each such sequence ${\cal S}$ do the following:
\begin{itemize}
\item[2.2.2.2.1.] Let $G'$ be the resulting triangulation after performing ${\cal S}$ starting at $\alpha$.
\item[2.2.2.2.2.] If $\ell =t$ and $G' = H$ then accept the instance and exit the (whole) algorithm.
\end{itemize}
\end{itemize}
\end{itemize}
\end{itemize}
\step Reject the instance.
\end{codeblock}
\end{algorithmnew}

\renewcommand{\figurename}{Fig.}
\setcounter{figure}{4}

We conclude with the following theorem:

\begin{theorem}\label{thm:final}
Given an instance $(G, H, k)$ of {\sc Combinatorial Flip Distance}, we can decide the instance correctly in time $\Oh(n+k\cdot c^{k}\cdot \lg{n})$, where $n=|V(G)|$.
\end{theorem}

\begin{proof}
First, note that an embedding of a triangulated graph can be computed in $\Oh(n)$ time~\cite{tarjanpt}, and we can determine the changed edges in $\Oh(n)$ time as well. The analysis of \algocomb~is the same as that for \algo, except for an $\Oh(\lg{n})$ multiplicative factor in the running time along each root-leaf path in the search tree. This multiplicative factor is for checking whether each of the at most $k$ flips performed along a path is admissible. Checking if flipping an edge $e$ is admissible amounts to checking whether the endpoints of the diagonal of $Q_e$ that is different from $e$ are adjacent or not. Brodal and Fagerberg~\cite{dynamic} describe a data structure for maintaining dynamic sparse graphs, which include planar graphs, that can answer adjacency queries in $\Oh(1)$ worst-case time per query, and that can perform insertions and deletions in amortized time $\Oh(1)$ and $\Oh(\lg{n})$, respectively, per operation. Using the data structure in~\cite{dynamic}, we can initially insert all edges of the graph into this data structure, which takes $\Oh(n)$ time. Afterwards, the running time along each root-leaf path in the search tree becomes $\Oh(k\cdot\lg{n})$. This is because along a root-leaf path in the search tree at most $k$ flips are performed, and the total running time along the path to test the admissibility of the $k$ flips is $\Oh(k)$, and for deleting the flipped edges and inserting the resulting edges is $\Oh(k\cdot\lg{n)}$. This completes the proof. \qed
%
\end{proof}

\section{Concluding remarks}\label{sec:conclusion}
We presented $\FPT$ algorithms for computing the flip distance between two triangulations, for several types of triangulations, that run in time $\Ostar(c^k)$, where $c \leq 2 \cdot 14^{11}$. Previously, only an $\Ostar(k^k)$-time $\FPT$ algorithm for the convex case was known~\cite{lucas}. There are several exciting research directions that ensue from our work:

\begin{itemize}
\item The running time of the presented algorithms is $\Ostar(c^k)$, where the constant $c \leq 2 \cdot 14^{11}$ is very large.  This constant can be improved at the expense of a more complicated analysis, but remains large. Whether an $\Ostar(c^k)$-time algorithm for a small $c$ exists is worth investigating.

\item The results in this paper imply exponential-size kernels for {\sc Flip Distance} and \CFD. The question of whether there is a polynomial-size kernel is challenging and remains open. (Recall that a kernel of size $2k$ was given by Lucas~\cite{lucas} for \FD~ in the convex case.)

\item The classical complexity of the problem of computing the flip distance between two triangulations remains open for triangulated graphs, labeled triangulated graphs, and triangulations of a convex polygon. Resolving the complexity of this problem for any of the aforementioned triangulation types, and in particular for triangulations of a convex polygon, is a very challenging open problem. We note that the problem of computing the flip distance between two triangulations of a point set and that of computing the flip distance between two labeled triangulated graphs, share that there is an identification between the vertices of the two triangulations, but differ in the requirement for an admissible flip. On the other hand, the problem of computing the flip distance between labeled triangulated graphs and that between triangulated graphs share the requirement for an admissible flip, but differ in that in the first problem an identification/bijection between the vertices is known whereas in the second problem it is not. Therefore, the three aforementioned problems exhibit three shades of variations, with the problem of computing the flip distance between labeled triangulated graphs occupying an intermediate position.

\item Despite that deciding whether two (unlabeled) triangulated graphs are isomorphic can be done in linear time~\cite{hopcroft}, we could not extend the results in this paper to triangulated graphs. The main obstacle was determining the ``changed'' edges. Investigating whether the same approach works for triangulated graphs is very interesting.
\end{itemize}

The results presented in this paper, and in particular, the derived representation and properties of a solution to a problem instance, reveal a lot of the structural properties of the flip distance problem. We hope that these results will shed more light on the problem that will help resolve some of the open questions described above.

 Finally, we note that {\sc Flip Distance} and \CFD~fall broadly into the category of reconfiguration problems, for which several parameterized complexity results have recently appeared  (see~\cite{mouadetal2,mouadetal1,mouadetal3}).

\bibliographystyle{plain}
\bibliography{ref}

\end{document}